\documentclass[acmsmall,nonacm]{acmart}
\usepackage{dsfont}
\usepackage{bm}
\usepackage{graphicx}
\usepackage{pgf}
\usepackage{caption}
\usepackage{hyperref}
\usepackage[normalem]{ulem}
\usepackage{booktabs}
\AtBeginDocument{%
  \providecommand\BibTeX{{%
    \normalfont B\kern-0.5em{\scshape i\kern-0.25em b}\kern-0.8em\TeX}}}

\ccsdesc[500]{Mathematics of computing~Markov processes}
\ccsdesc[300]{Networks~Network performance analysis}

\newtheorem{claim}{Claim}
\newtheorem*{claim*}{Claim}

\theoremstyle{definition}
\newtheorem{remark}{Remark}
\usepackage{natbib}
\usepackage{enumitem}
\allowdisplaybreaks
\newcommand{\Exp}{\mbox{Exp}}

\renewcommand{\Pr}{\mathbb{P}} 
\newcommand{\expect}{\mathbb{E}} 
\newcommand{\pup}{p^{(\numserver)}} 
\newcommand{\PrQup}{P^{(\numserver)}_Q} 
\newcommand{\numserver}{n} 
\newcommand{\numtype}{K} 
\newcommand{\resource}{m} 
\newcommand{\resourceup}{\resource^{(\numserver)}}
\newcommand{\resourcemax}{\resourceup_{\max}}
\newcommand{\type}{i} 
\newcommand{\lambdasup}{\lambda^{(\numserver)}}
\newcommand{\load}{\rho}
\newcommand{\loadup}{\load^{(\numserver)}}
\newcommand{\numjob}{X} 
\newcommand{\numjobup}{\numjob^{(\numserver)}} 
\newcommand{\numjobvec}{\bm{\numjob}} 
\newcommand{\numjobvecup}{\numjobvec^{(\numserver)}} 
\newcommand{\scalenumjob}{\widetilde{X}} 
\newcommand{\scalenumjobvec}{\widetilde{\bm{X}}} 
\newcommand{\scalenumjobup}{\scalenumjob^{(\numserver)}} 
\newcommand{\scalenumjobvecup}{\scalenumjobvec^{(\numserver)}} 
\newcommand{\numjoblim}{y} 
\newcommand{\numjoblimvec}{\bm{\numjoblim}} 
\newcommand{\numjoblimup}{y^{(\numserver)}} 
\newcommand{\Qup}{Q^{(\numserver)}} 
\newcommand{\Qvecup}{\bm{Q}^{(\numserver)}} 

\newcommand{\stateup}{\bm{U}^{(\numserver)}} 
\newcommand{\numjoba}{Y} 
\newcommand{\numjobaup}{\numjoba^{(\numserver)}}
\newcommand{\numjobaupscaled}{\widetilde{\numjoba}^{(\numserver)}}
\newcommand{\numjobavec}{\bm{\numjoba}}
\newcommand{\numjobavecup}{\numjobavec^{(\numserver)}}

\global\long\def\mU{\mathcal{U}}%

\title{Zero Queueing for Multi-Server Jobs}
\date{}

\begin{document}

\author{Weina Wang}
\email{weinaw@cs.cmu.edu}
\affiliation{%
  \institution{Carnegie Mellon University}
  \streetaddress{Computer Science Department, 5000 Forbes Ave}
  \city{Pittsburgh}
  \state{PA}
  \country{USA}
  \postcode{15213}
}

\author{Qiaomin Xie}
\email{qiaomin.xie@cornell.edu}
\affiliation{%
  \institution{Cornell University}
  \streetaddress{School of Operations Research and Information Engineering, 237 Frank H.T.\ Rhodes Hall}
  \city{Ithaca}
  \state{NY}
  \country{USA}
  \postcode{14853}
}

\author{Mor Harchol-Balter}
\email{harchol@cs.cmu.edu}
\affiliation{%
  \institution{Carnegie Mellon University}
  \streetaddress{Computer Science Department, 5000 Forbes Ave}
  \city{Pittsburgh}
  \state{PA}
  \country{USA}
  \postcode{15213}
}

\renewcommand{\shortauthors}{Weina Wang, Qiaomin Xie, \& Mor Harchol-Balter}

\begin{abstract}
Cloud computing today is dominated by \emph{multi-server jobs.}  These are jobs that request multiple servers \emph{simultaneously} and hold onto all of these servers for the duration of the job.   Multi-server jobs add a lot of complexity to the traditional one-server-per-job model: an arrival might not ``fit'' into the available servers and might have to queue, blocking later arrivals and leaving servers idle. From a queueing perspective, almost nothing is understood about multi-server job queueing systems; even understanding the exact stability region is a very hard problem.

In this paper, we investigate a multi-server job queueing model under scaling regimes where the number of servers in the system grows.  Specifically, we consider a system with multiple classes of jobs, where jobs from different classes can request different numbers of servers and have different service time distributions, and jobs are served in first-come-first-served order.  The multi-server job model opens up new scaling regimes where \emph{both} the number of servers that a job needs and the system load scale with the total number of servers. Within these scaling regimes, we derive the first results on stability, queueing probability, and the transient analysis of the number of jobs in the system for each class. In particular we derive sufficient conditions for zero queueing. Our analysis introduces a novel way of extracting information from the Lyapunov drift, which can be applicable to a broader scope of problems in queueing systems.
\end{abstract}
\maketitle

\section{Introduction}\label{sec:introduction}

Queueing theorists have long been interested in characterizing the probability that an arriving job will have to queue, i.e., the \emph{queueing probability}.
This question has a rich history where it has been investigated under different scaling regimes.  Consider the classical M/M/$\numserver$ system that has load $1-\beta \numserver^{-\alpha}$ with $0<\beta< 1$ and $\alpha\ge 0$.  The case of $\alpha=1/2$ is known as the Halfin-Whitt regime \cite{HalWhi_81} and serves as the critical threshold that separates diminishing queueing probability from non-diminishing queueing probability, as the number of servers $\numserver$ grows. 
In particular, it is known that the queueing probability is diminishing when $0\le \alpha < 1/2$ (sub-Halfin-Whitt); is strictly between $0$ and $1$ when $\alpha=1/2$ (Halfin-Whitt); and goes to $1$ when $\alpha>1/2$ (see, e.g., \cite{BraDaiFen_17}).

Today's computing clusters are more general than the M/M/$\numserver$ model used in the past. In particular, the jobs typically request {\em multiple servers simultaneously} and hold onto them for the duration of the job.   This difference is largely a result of machine learning jobs like TensorFlow
\cite{AbaBarChe_16,lin2018model}, which are highly parallel.  For example, when we look at 
Google's Borg Scheduler \cite{verma2015large}, we see that the number of servers occupied by an individual job can be anywhere from $1$ to $100000$, as illustrated in Figure~\ref{fig:multipleGoogle}~\cite{Wilkes19, Eurosys20}.  We refer to jobs that occupy multiple servers/cores as {\em multi-server jobs}.  While multi-server jobs have always existed in the niche supercomputing world, they have now become mainstream.   

\begin{figure}[h]
    \begin{center}
    \includegraphics[scale=0.3]{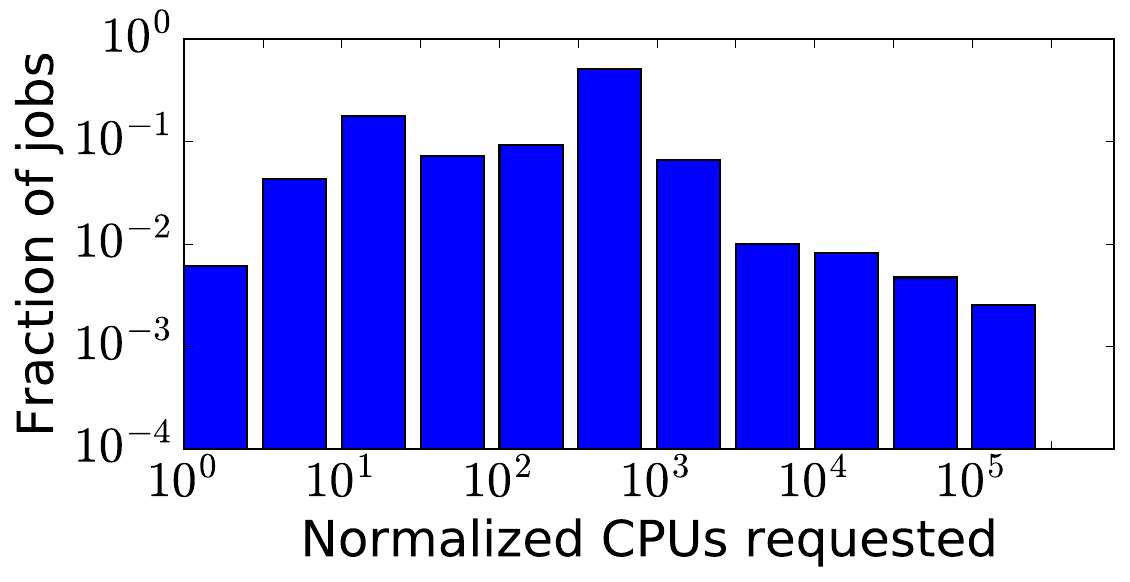}
    \end{center}
    \caption{Distribution of the number of CPU cores requested by individual Google jobs according to Google's recently published Borg Trace \cite{Eurosys20, Wilkes19} as taken from \cite{StabilityMultiserver20}.  For privacy reasons, Google publishes only normalized numbers, however it is still clear that the range spans $5$ orders of magnitude across different jobs.}
    \label{fig:multipleGoogle}
\end{figure}

The advent of multi-server jobs requires us to generalize our queueing models.
Figure~\ref{fig:multiserverMGk} shows an illustration of what we will refer to as {\em the multi-server job queueing model}.  Here there are a total of $\numserver$ servers.  Jobs arrive with average rate $\lambda$ and are served in First-Come-First-Serve (FCFS) order~\footnote{Importantly, we note that we are assuming that jobs arrive to a centralized queue and are served in FCFS order, rather than trying to ``pack" jobs into servers.  A centralized FCFS scheduler is the default scheduler used in the cloud-computing industry when running multi-server jobs \cite{verma2015large,Eurosys20}. Even when there are multiple priority classes of jobs, as in \cite{Eurosys20}, within each class the jobs are served in FCFS order.
}.
With probability $p_i$ an arrival is of class $i$.  Each arriving job of class $i$ requests $m_i$ servers and holds onto these servers for time distributed according to random variable $S_i$. We will refer to the number of servers a job demands as the \emph{server need} of the job, while we refer to the time that the job holds onto the servers as its {\em service time}.

\begin{figure}[hbtp]
\begin{center}
    \includegraphics[scale=0.47, clip, trim = 0 0 0 0]{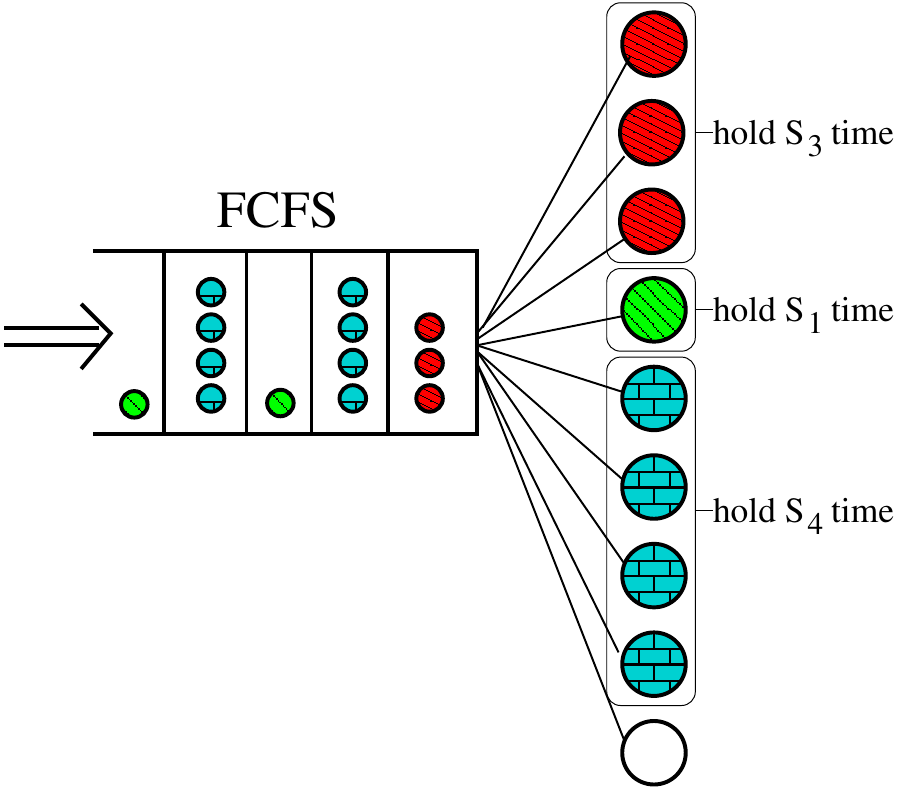}
    \end{center}
\caption{The multi-server job queueing model with $\numserver =9$ servers.  An arriving job of class $i$ requests $m_i$ servers and holds onto these servers for time distributed according to random variable $S_i$.  In this particular illustration, $m_i = i$. Note that the job at the head of the queue cannot enter service because it does not fit; hence one server is idle.}
\label{fig:multiserverMGk}
\end{figure}

The multi-server job queueing model opens up new scaling regimes where both the server needs of jobs and the system load scale with the number of servers.  In this paper we ask:
\begin{quote}
    {\em Under this new joint scaling regime, when is diminishing queueing probability achievable?}
\end{quote}
The property of \emph{diminishing queueing probability} is also referred to as (asymptotically) \emph{zero queueing} in the literature, and we will use these two terms interchangeably.
The answer to this question can take a very different form from the classical results.
To see this, let us consider the following simple example.  Suppose jobs arrive to a system with $\numserver$ servers according to a Poisson process of rate $\lambda$ where all jobs are of the {\em same} class 1.   Suppose that $S_1 \sim \Exp(\mu_1)$ and $\resource_1 =\numserver/2$ servers.
Even in this degenerate version of the multi-server job model, the work associated with a job depends both on its service time, $S_1$, and also on its server need, $\resource_1$.  Therefore, we redefine the notion of load to be $\rho\triangleq\frac{\lambda \resource_1 \expect[S_i]}{\numserver}=\frac{\lambda}{2\mu_1}$.  Then even when the load $\load$ is a positive constant smaller than~$1$, analogous to the classical subcritical regime in a large system \cite{Igl_73}, the queueing probability does not diminish when $\numserver$ grows.  The reason is that the large server need of jobs makes the system equivalent to an M/M/$2$ system, which is effectively a small system even though the actual number of servers grows.

Understanding the queueing probability becomes much more challenging when there are multiple classes of jobs that have heterogeneous server needs and service time distributions. Here it is often the case that some servers remain idle, because the job at the head of the queue cannot ``fit'' into the remaining unused servers, as illustrated in  Figure~\ref{fig:multiserverMGk}. Although there may exist jobs in the queue with smaller server needs, those jobs are blocked by the head-of-queue job.  
Due to this \emph{head-of-queue blocking}, the process that tracks the number of each class of jobs in service is highly complicated. Almost nothing is known on the performance of the
multi-server job queueing model. Attempts to derive the
steady-state distribution have assumed
highly simplified systems with only $\numserver = 2$ servers \cite{BrilGre_84, FilippopoulosKaratza07}, where solutions are already highly complex, involving roots to a quartic equation.
Even characterizing the stability region of the system is an open problem except for the special cases where all jobs have the same service rates  \cite{RumyantsevMorozov17, MorozovRumyantsev16, Afanasevaetal19}, or where there are only two job classes with different service rates \cite{StabilityMultiserver20}.

\subsection*{Our results}
We consider a system of $\numserver$ servers and $\numtype$ classes of jobs.  For system parameters that are functions of $\numserver$, we add a superscript $^{(\numserver)}$ to indicate the dependency unless otherwise specified. Jobs arrive according to a Poisson process with rate $\lambdasup$ and an arrival is of class $\type$ with probability $\pup_{\type}$. Then for each job class $\type$ with $1\le\type\le \numtype$, the arrivals form a Poisson process with rate $\lambdasup_\type:=\lambdasup p^{(n)}_\type.$ 
Jobs of class $\type$ have i.i.d.\ service times exponentially distributed with rate $\mu_{\type}$, which does not scale with $\numserver$.  Each class $\type$ job has a server need of $\resourceup_\type$ servers.  Let $\resourcemax=\max_{1\le\type\le\numtype}\resourceup_{\type}$ denote the maximum server need.

\subsubsection*{\textbf{Stability result}}
We first define a notion of \emph{load} for multi-server job queueing systems,
which will be used in the stability condition.
Traditionally, for single-server jobs, the \emph{work} brought in by a job is quantified by its service time, i.e., the time the job needs to occupy a server.  However, for multi-server jobs, we need to account for the fact that they occupy multiple servers simultaneously.  Therefore, we quantify the \emph{work} brought in by a multi-server job using the product ``server need $\times$ service time'', which is consistent with the CPU-hours metric used in practice \cite{Eurosys20}.  Then we define the load of class~$\type$ jobs to be $\loadup_{\type}=\frac{\lambdasup_{\type} \resourceup_{\type}}{\numserver\mu_{\type}}$, and the total system load to be $\loadup=\sum_{\type=1}^{\numtype}\loadup_{\type}$.

We then show in Theorem~\ref{thm:stability} that a sufficient condition for the system to be stable is that
\begin{equation}\label{eq:stability-intro}
\loadup < 1-\frac{\resourcemax}{\numserver}.
\end{equation}
Note that when the load $\loadup > 1$, no policy can stabilize the system.  Therefore, the condition $\loadup < 1-\frac{\resourcemax}{\numserver}$ is \emph{asymptotically tight} when $\resourcemax = o(\numserver)$.
\footnote{We use the standard asymptotic notation: $f(n)=o(g(n))$ if $\lim_{n\to\infty} f(n)/g(n)=0$; $f(n)=\Theta(g(n))$ if $\lim_{n\to\infty} f(n)/g(n)$ is a positive constant; $f(n)=\Omega(g(n))$ if $\lim_{n\to\infty} f(n)/g(n)$ is no smaller than a positive constant.}
We comment that this sufficient condition is in general not tight non-asymptotically, and finding the exact stability region is a hard open problem.  We refer the readers to the section on related works (Section~\ref{sec:related-work}) for more details.

\subsubsection*{\textbf{Queueing probability}}
The most important result in this paper is our characterization of the queueing probability, denoted by $\PrQup$, which refers to the probability that a job cannot enter service immediately upon arrival in steady state.  Specifically, we chart the joint scaling regimes for the server needs of jobs and the system load to answer the question of when {\em diminishing
queueing probability}
is achievable, i.e., when $\PrQup\to 0$ as $\numserver\to\infty$, for the multi-server job queueing model.

For ease of exposition,
we parameterize the scaling regimes as the number of servers $\numserver$ grows in the following way.
Assume that the load $\loadup=1-\beta \numserver^{-\alpha}$ with $0<\beta<1$ and $\alpha\ge 0$ and the maximum server need $\resourcemax=\Theta(\numserver^{\gamma})$ with $0\le\gamma\le 1$.  We focus on the key parameters $\alpha$ and $\gamma$ and divide the regimes into four regions as illustrated in Figure~\ref{fig:scaling-regimes}.
We will start our discussion with results in simpler regions, Regions~1--3, since they are closely connected to classical queueing models.  Then we will turn to results in the \emph{main region of focus}, Region~4, which corresponds to the novel scaling regime where the load and the server needs scale \emph{jointly}.

\begin{figure}
\centering
\includegraphics[scale=0.5]{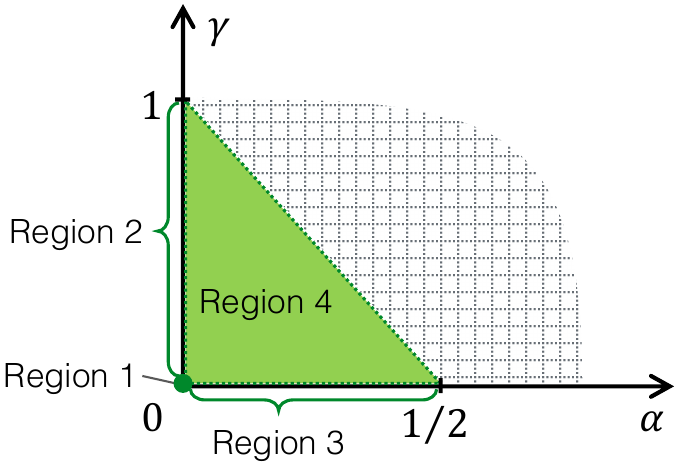}
\caption{Joint scaling regimes for system load $\loadup=1-\beta \numserver^{-\alpha}$ and server need $\resourcemax=\Theta(\numserver^{\gamma})$.}
\label{fig:scaling-regimes}
\vspace{-0.1in}
\end{figure}

\begin{itemize}[leftmargin=1em,itemsep=8pt]
\item \textbf{Region 1:} When the load $\loadup$ stays \emph{constant} ($\alpha=0$) and the server needs also stay \emph{constant} ($\gamma=0$), does the queueing probability, $\PrQup$, diminish as $\numserver\to\infty$?

\emph{Our result:} We show that $\PrQup$ is indeed diminishing in this region.  This result should not be surprising given the classical result that when jobs have unit server need, the queueing probability is diminishing under constant load (see, e.g., \cite{HalWhi_81}, where a special case of their results analyzes a system where each job has a server need of one).

\item \textbf{Region 2:} When the load $\loadup$ stays \emph{constant} ($\alpha=0$) but the server needs \emph{scale} ($\gamma>0$), does the queueing probability, $\PrQup$, diminish as $\numserver\to\infty$?

\emph{Our result:} We show that $\PrQup$ is diminishing when $\gamma<1$, i.e., when a job drawn from the class of highest server need occupies a diminishing fraction of the servers during service.  The threshold $\gamma<1$ is tight in the sense that for a single job-class system, $\PrQup$ is not diminishing when $\gamma=1$.  This is easy to see by noting that when $\gamma=1$, the single job-class system is equivalent to a classical M/M/$s$ system where $s=\lfloor\numserver/\resourcemax\rfloor$ is a constant and the load is also a constant.

\item \textbf{Region 3:} When the load $\loadup$ is in \emph{heavy-traffic} ($\alpha>0$) and the server needs stay \emph{constant} ($\gamma=0$), does the queueing probability, $\PrQup$, diminish as $\numserver\to\infty$?

\emph{Our result:} We show that $\PrQup$ is diminishing when the load satisfies $0<\alpha<1/2$, which is a traffic regime analogous to the sub-Halfin-Whitt regime in the literature.  The threshold  $\alpha<1/2$ is tight in the sense that for a single job-class system, $\PrQup$ is not diminishing when $\alpha\ge 1/2$.  Again, this can been seen by noting that when $\alpha\ge 1/2$, the single job-class system is equivalent to a classical M/M/$s$ system with $s=\lfloor\numserver/\resourcemax\rfloor=\Theta(\numserver)$, whose traffic is at least as heavy as the Halfin-Whitt regime.

\item \textbf{Region 4:} When the load $\loadup$ is in \emph{heavy-traffic} ($\alpha>0$) and the server needs also \emph{scale} ($\gamma>0$), does the queueing probability, $\PrQup$, diminish as $\numserver\to\infty$? 

\emph{Our result:} We show that $\PrQup$ is diminishing when $2\alpha+\gamma<1$.  This exact characterization of the joint scaling regime for diminishing queueing probability formalizes the intuition that queueing is negligible in large systems if each job does not need too many servers and the load is not too heavy.
This threshold $2\alpha+\gamma<1$ is also tight in the sense that for a single job-class system, the queueing probability is not diminishing when $2\alpha+\gamma\ge 1$.  To see this, we again consider the equivalent M/M/$s$ system with $s=\lfloor\numserver/\resourcemax\rfloor$.  It is not hard to verify that the load of this equivalent system is $1-\Theta(\numserver^{-\alpha})=1-\Theta(s^{-\frac{\alpha}{1-\gamma}})$, which is at least as heavy as Halfin-Whitt since $\frac{\alpha}{1-\gamma}\ge1/2$ under the condition that $2\alpha+\gamma\ge 1$.
\end{itemize}

In fact, our characterization of the queueing probability in Theorem~\ref{thm:diminishing-queueing} has the following general form, which not only unifies all of the four regions above but also provides an upper bound on the rate of convergence for the diminishing queueing probability:
\begin{equation}\label{eq:queueingprob-intro}
\PrQup\le \frac{1}{1-\loadup}\left(3\numtype\sqrt{\frac{\resourcemax}{\numserver}}+\frac{\resourcemax}{\numserver}\right).
\end{equation}

\subsubsection*{\textbf{Transient analysis for the number of jobs in the system}}
We also characterize the transient behavior of the system under constant load in the limit as the number of servers $\numserver$ goes to infinity. Specifically, in Theorem~\ref{THM:FINITE_TIME} we show that the sequence of processes for appropriately scaled number of jobs in the system converges to a deterministic system over any finite time interval, as $\numserver \rightarrow \infty.$  In particular, the deterministic system is the unique solution to a fluid model, which converges to an equilibrium point as $t\rightarrow \infty.$ Therefore, we have established that, for any sufficiently large {\em fixed-time} interval, the scaled number of jobs of a large system (large $n$) is approximated by the equilibrium point of the fluid model. We conjecture that the stationary distribution converges to the equilibrium point.

\subsection*{\textbf{Technical challenges and our new approach}}
To better understand the queueing dynamics, it is crucial to characterize how many jobs of each class are in service, since that determines the total departure rate of jobs.  However, due to the heterogeneous server needs of jobs, the number of jobs in service turns out to be very hard to analyze.
The root of the difficulty is that the process by which jobs enter service can be quite bursty.   For example, when a job in service finishes, it could happen that \emph{many} jobs in the queue will be admitted into service \emph{all at the same time} if the completed job has a high server need; or it could happen that \emph{no job} in the queue will enter service at all if the completed job frees up a small number of servers and the head-of-queue job cannot fit.  This ``jitter effect'' makes it highly challenging to reason about jobs in service.

Our approach is similar in spirit to the drift method (see Section~\ref{sec:proof-diminishing}).  However, our construction of the Lyapunov function cannot be derived using existing methods as in \cite{LiuYin_20, LiuGonYin_20, LiuYin_19, WenWan_21, WenZhoSri_21,  ErySri_12, MagSri_16, WanMagSri_18} and instead requires new insights.  In particular, we find a Lyapunov function whose drift has an interesting upper-bounding function, which then allows us to establish an upper bound on the queueing probability.  
Our Lyapunov function also has an intuitive meaning: it corresponds to the  total work contributed by all the jobs in the system.
Here the work contributed by a job is the product of its server need and its expected service time. More details on the drift method and the steps in our approach are given in Sections~\ref{subsec:background} and \ref{subsec:tech-diminishing}.

\section{Related Works}\label{sec:related-work}
\subsubsection*{\textbf{Multi-server job models}}

At present almost nothing is known regarding the performance of the
multi-server job model.  A few works have sought to derive the
steady-state distribution of the number of jobs in the system under a 
highly simplified model, where all jobs have the same exponential
service duration $S_{\type}\sim \Exp(\mu)$, for all classes $\type$, and there are only
$\numserver=2$ servers, see papers by \citet{BrilGre_84} and \citet{FilippopoulosKaratza07}.  However these solutions are highly complex, typically involving roots to a quartic equation; this makes the 
solutions impractical.  Earlier work by \citet{Kim79}, based on 
a matrix analytic approach, is similarly 
impractical since it scales exponentially with the size of the system.  A survey paper by Melikov \cite{Melikov96} summarizes these approaches and a few others.  
In summary, understanding the response time for multi-server job systems with more than $\numserver=2$ servers is entirely open,
even when all jobs have the same exponentially-distributed service durations.

Not only is the response time intractable for the multi-server job model, but even the stability region for this model (under FCFS scheduling) is only partially understood. In 2017, \citet{RumyantsevMorozov17} derived the stability region for the multi-server job model where, for all class~$\type$, all 
jobs have the {\em same} exponential service duration $S_\type \sim \Exp(\mu)$.  This work was generalized in \cite{MorozovRumyantsev16} and \cite{Afanasevaetal19} to allow for more general arrival processes, still under the assumption that all jobs have the same exponential service duration.  Very recently, \citet{StabilityMultiserver20} derived a simple closed-form expression for the stability region for the multi-server job model where jobs have {\em different} exponential service time durations. Unfortunately, the work \cite{StabilityMultiserver20} is limited to the case of only two classes.  The characterization of stability in the case of more than two job classes is open.

We reiterate that the key technical difficulty is rooted in the heterogeneous server needs of multi-server jobs, rendering classical general frameworks for stability less effective or inapplicable.
The sufficient condition \eqref{eq:stability-intro} that we establish for stability is derived under the framework based on Lyapunov drift (see, e.g., \cite{SriYin_14}, for a comprehensive coverage). Here the heterogeneity in server needs makes it hard to tighten this condition in the non-asymptotic regime.  Another general stability framework is based on the saturation rule \cite{BacFos_95} for systems with certain monotonicity.  Roughly speaking, the monotonicity property in this framework requires that when the queueing system has a finite number of job arrivals, delaying job arrival times can only delay the time when all the jobs are completed.  The monotonicity property is satisfied by many classical queueing models including the M/M/$\numserver$ model.  However, it is not satisfied by the multi-server job model under FCFS.  Counterexamples can be constructed where delaying job arrivals actually leads to a shorter overall completion time due to less server idling time.

While very little is known about the performance of multi-server job
models, there is a close cousin of the model, called the {\em dropping
model} which is analytically tractable under very general settings.  In the dropping
model, those jobs which cannot immediately receive service are simply dropped.  When the job durations are exponentially distributed, the stationary distribution of the dropping model exhibits a beautiful product form.
\citet{ArtKau_79} were the first to observe the product
form.  \citet{Whi_85} generalized the model to allow jobs to demand multiple resource types (e.g., both CPU and I/O).  \citet{vanDijk89}
allowed durations to be generally distributed.  \citet{Tikhonenko05}
combined aspects of \cite{Whi_85} and \cite{vanDijk89}.

The multi-server job model is also related to \emph{streaming models} for communication networks.  Here the resource being shared is bandwidth in the network.  The ``jobs'' are audio or video flows which require a fixed bandwidth reservation to run (this is akin to needing a fixed number of servers).  Flows requiring fixed bandwidth are often referred to as streaming (see \cite{BenBenDel_01}), but are also sometimes referred to as ``inelastic jobs'' (see \cite{PonKimMel_10, Mel_96}).  The papers dealing with streaming jobs typically operate in the dropping model, where the goal is to schedule to minimize a cost related to dropping probabilities (see \cite{DasSri_99,HunLaw_97,BeaGibZac_95,HunKur_94}).  Note that the setting here is more complex than the multi-server job dropping model.  The added complexity sometimes comes from the fact that the authors are seeking an optimal dropping policy and sometimes is due to a network setting.

Another class of related models are models for the virtual machine (VM) scheduling problem \cite{MagSri_2013,MagSriYin_14,XieDonLu_15,PsyGha_18,PsyGha_19}, but again only limited results are available for job response times and most works focus on stability.  In the VM scheduling problem, a VM job requests multiple units of resource such as CPUs.  But the server model in VM scheduling is different from the multi-server job model we study in this paper.  In VM scheduling, a server has several CPUs and usually can accommodate multiple jobs, but a job cannot be spread across multiple servers.

\subsubsection*{\textbf{Diminishing queueing probability in other models}}
The concept of diminishing queueing probability has been studied across a wide class of problems.   As already discussed, in the M/M/$\numserver$ model, it was shown that diminishing queueing probability occurs in the sub-Halfin-Whitt regime.
Recently, motivated by the desire for low latency in today's computing applications, the interest in diminishing queueing probability (and queueing time) have been greatly renewed.  Investigating conditions and policies for achieving diminishing queueing probability has become a rapidly growing research area with a rich body of work.

The queueing probability (and queueing time) under scaling regimes has also been studied in load-balancing models, where each server has its own queue, and where each arriving job is immediately dispatched to a server upon arrival.  The Join-the-Idle-Queue (JIQ) policy routes every arriving job to an idle queue, if one exists, otherwise to a randomly selected queue \cite{LuXieKli_11,Sto_15}. For JIQ it was shown that diminishing queueing probability can be achieved when the load, $\loadup$, is lighter than $1- \frac{1}{\numserver}$ \cite{LiuYin_20,LiuGonYin_20,Liu_19}.   Another example is the Power-of-$d$-Choices policy (Po$d$), where every arriving job samples $d$ queues and goes to the shortest of these \cite{VveDobKar_96,Mit_96}.   For Po$d$ it was shown that diminishing queueing probability can be achieved when $d=\Omega\left(\frac{\log\numserver}{1-\loadup}\right)$ in the sub-Halfin-Whitt regime and when $d=\Omega\left(\frac{\log^2\numserver}{1-\loadup}\right)$ for $\loadup$ lighter than $1- \frac{1}{\numserver}$ \cite{LiuYin_20,LiuGonYin_20,LiuYin_19}.
Moreover, the Join-the-Shortest-Queue (JSQ) load-balancing policy can be viewed as a special case of Po$d$, where $d=\numserver$.  Thus it too has diminishing queueing probability.
Finally, \citet{WenZhoSri_21} extend the traditional load-balancing model to address data locality and again prove diminishing queueing probability under specific conditions on the data locality.

Another model where diminishing queueing time has been investigated is the multi-task job model.  Here, again, each server has its own queue.  However, each job consists of $k$ tasks that can run on servers in parallel with i.i.d.\ service time requirements, but the job is not completed until all of its tasks complete \cite{WenWan_21}.  The multi-task model is closest to our own, but differs in several ways.   Firstly, there are queues at each server, where the tasks of a job need to be dispatched to the queues upon arrival.  This implies, importantly, that the tasks of a job typically do not end up  executing simultaneously.  Secondly, the multi-task model is different from our own in that the individual tasks of a job have i.i.d.\ service time requirements which may be quite different from each other.  Within the multi-task model, \citet{WenWan_21} analyze a Batch-Filling-$d$ policy, where each arriving job samples $kd$ queues and fills its tasks into queues so as to minimize the individual task queueing times.
For this Batch-Filling-$d$ algorithm, \citet{WenWan_21} show that diminishing queueing time for jobs can be achieved in the sub-Halfin-Whitt regime when $d$ is sufficiently high.

\section{Model}\label{sec:model}

The basic description of the multi-server job queueing model with $\numserver$ servers  follows Figure~\ref{fig:multiserverMGk}.  The notation that we have discussed so far appears in Table~\ref{table:variables_job}.
We define $\resourcemax:= \max_{1\leq \type \leq \numtype} \resourceup_\type$ as the maximum server need, and $\loadup:=\sum_{\type=1}^{\numtype}\loadup_\type =\sum_{\type=1}^{\numtype}\frac{\lambdasup_\type \resourceup_\type}{\numserver \mu_\type}$ as the system load.
Let $\mu_{\max}:=\max_{1\leq \type \leq \numtype} \mu_\type.$

\begin{table}
\hspace{-0.3in}
\begin{minipage}[c]{.45\textwidth}
\centering
\begin{tabular}{rl}
\toprule
Service time & $S_\type \sim\text{Exp}(\mu_{\type})$\\
Server need & $\resourceup_{\type}$\\
Arrival rate & $\lambdasup_{\type}:=\lambdasup\cdot \pup_{\type}$\\
Class~$\type$ load &  $\loadup_{\type}:=\frac{\lambdasup_{\type} \resourceup_{\type}}{\numserver\mu_{\type}}$\\
\bottomrule
\end{tabular}
\caption{Variables related to class $\type$ jobs}
\label{table:variables_job}
\end{minipage}
\hspace{0.1in}
\begin{minipage}[c]{.45\textwidth}
\centering
\begin{tabular}{rl}
\toprule
Number of servers & $\numserver$\\
Number of class $\type$ jobs in system & $\numjobup_{\type}$\\
Number of class $\type$ jobs in queue &  $\Qup_\type$\\
Classes of jobs in system & $\stateup$\\
\bottomrule
\end{tabular}
\caption{Variables related to system states}
\label{table:variables_system}
\end{minipage}
\vspace{-0.3in}
\end{table}

Arriving jobs enter a First-Come-First-Served (FCFS) queue, which has an infinite capacity.  
 As illustrated in Figure~\ref{fig:multiserverMGk}, when a job arrives, it enters service if the queue is empty and the number of idle servers is at least the job's server need; otherwise, the job enters the queue to wait for service. 

The state of the multi-server job queueing system can be described by an ordered list of the classes of all jobs in the system, in arrival order. We denote the state at time $t$  by $\stateup(t)=\left(u_1(t),u_2(t),\dots,u_J(t)\right),$ where $u_j(t)$ is the class of the $j$th oldest job in the system at time $t$. 
For the state shown in Figure~\ref{fig:multiserverMGk}, assuming that the three jobs in service arrived in the order of $3$, then $1$, then $4$,  the state descriptor is given by $(3,1,4,3,4,1,4,1)$. 
Note that this state descriptor is \emph{ordered} and the order determines which jobs are in service.  Therefore, the total job departure rate depends on the order, which implies that the queue is not an order-independent queue \cite{Krz_11,BonCom_17}.
It can be verified that the continuous-time process, $\stateup=\bigl\{ \stateup(t),t \geq 0 \bigr\}$, forms an irreducible Markov chain taking values in $\mU = \{(u_1, u_2, \dots, u_J)\colon J\in\mathbb{N}, u_j \in \{1,2,\dots,\numtype\}\forall j\}$.

Let $\numjobup_\type(t)$ denote the number of class $\type$ jobs in the system at time $t$ for $i=1,\ldots,\numtype$. Let $\numjobvecup(t)$ be the corresponding vector. Note that the system size process $\numjobvecup=\bigl\{\numjobvecup(t),t\geq 0 \bigr\}$ is not a Markov chain. We define the queue size process $\Qvecup=\bigl\{\Qvecup(t),t\geq 0\bigr\}$, where $\Qvecup(t)=\bigl(\Qup_1(t),\ldots,\Qup_{\numtype}(t)\bigr)$
and $\Qup_\type(t)$ is the number of class $\type$ jobs waiting in the queue at time $t.$ 
We summarize the variables related to system state in Table~\ref{table:variables_system}.

Observe that each of $\numjobvecup$ and $\Qvecup$ is a function of the Markov chain $\stateup.$ When the Markov chain $\stateup$ is positive recurrent with a unique stationary distribution, $\numjobvecup$ also has a unique stationary distribution,  as does $\Qvecup$.
We use $\numjobvecup(\infty)$ and $\Qvecup(\infty)$ to denote the random variables that have the stationary distribution of $\numjobvecup$ and $\Qvecup$, respectively.

In this paper, we are interested in characterizing the queueing probability, under FCFS scheduling, in different scaling regimes. Here we use the term \emph{queueing probability}, denoted by $\PrQup$,  to refer to the steady-state probability that a job cannot enter service immediately upon arrival.
Note that when there are more than or equal to $\resourcemax$ idle servers, an arriving job will enter service immediately.  Thus,
we can upper bound $\PrQup$ as follows
\begin{align}
\PrQup\leq \Pr\left(\sum_{\type=1}^{\numtype}\resourceup_{\type}\numjobup_{\type}(\infty)> \numserver-\resourcemax\right). 
\end{align}

\section{Main Results}\label{sec:results}
In this section we formally present our main theorems.


\begin{theorem}[{\bf Stability Condition}]\label{thm:stability}
Consider the system with $\numserver$ servers and $\numtype$ classes of multi-server jobs.  Under the FCFS policy, the Markov chain $\stateup$ is positive recurrent, i.e., the system is stable, when the load $\loadup$ satisfies
\begin{equation}
\loadup < 1-\frac{\resourcemax}{\numserver}.
\end{equation}
\end{theorem}


\begin{theorem}[{\bf Diminishing Queueing Probability}]\label{thm:diminishing-queueing}
Consider the system with $\numserver$ servers and $\numtype$ classes of multi-server jobs.  Assume that it uses the FCFS policy and the load satisfies the stability condition $\loadup < 1-\frac{\resourcemax}{\numserver}$. Then the queueing probability is upper bounded by:
\begin{equation}\label{eq:queueing-upper-bound}
\PrQup\le \frac{1}{1-\loadup}\left(3\numtype\sqrt{\frac{\resourcemax}{\numserver}}+\frac{\resourcemax}{\numserver}\right).
\end{equation}
Consequently, if the joint scaling of $\loadup$ and $\resourcemax$ satisfies $\frac{1}{1-\loadup}\sqrt{\frac{\resourcemax}{\numserver}}=o(1)$, then the queueing probability is diminishing, i.e.,
\begin{equation}\label{eq:queueing-prob-diminish}
\lim_{\numserver\to\infty}\PrQup=0.
\end{equation}
\end{theorem}

Finally, we show that the evolution of the scaled number of jobs of each class in the system, $\frac{1}{\lambdasup_\type}\numjobup_\type(t)$, can be approximated by a deterministic system, over any finite time horizon $[0,T].$ In particular, the deterministic system is governed by the following differential equations:
\begin{equation}\label{eq:ODE0}
\dot{\numjoblim}_\type(t) = 1 - \min\left\{\mu_\type\numjoblim_\type(t),\frac{1}{\load_{\type}}\right\},\qquad \type\in\{1,\ldots,\numtype\},
\end{equation}

Note that when each job can enter service immediately upon arrival, the original system has the same dynamics as an ``enlarged'' system where each class of jobs has access to a separate set of $\numserver$ servers and is served in FCFS order, i.e., class $\type$ jobs run as an M/M/$\frac{\numserver}{\resourceup_\type}$ queue with arrival rate $\lambdasup_\type$ and service rate $\mu_\type$. The diminishing queueing probability implied by Theorem~\ref{thm:diminishing-queueing} suggests that the original system can be approximated by such an enlarged system. On the other hand, we can view the solution $\numjoblim_\type(t)$ to equation~(\ref{eq:ODE0}) as a deterministic approximation to the sample path of the scaled number of class $\type$ jobs in the enlarged system. The deterministic system $\numjoblimvec(t)$ hence also provides an approximation for the scaled number of jobs in the original system.

\begin{theorem}[{\bf Transient behavior of the number of jobs in the system}]
\label{THM:FINITE_TIME} Suppose that for each class $\type\in \{1,\ldots,\numtype\}$, the load satisfies $\loadup_\type=\load_\type>0$ for all $\numserver$,  and $\load=\sum_{\type=1}^{\numtype}\load_\type<1$. 
Assume that $\lim_{\numserver\to\infty}\frac{1}{\lambdasup_\type}\numjobup_\type(0)=\numjoblim_{\type}(0)$
in probability for each $\type$, where $\numjoblimvec(0)$
is a deterministic initial condition such that $0\leq \numjoblim_\type(0)<\frac{1}{\load\mu_\type}$ for all $\type$.
Let $\numjoblimvec(t)$ be the unique solution to the differential equation~\eqref{eq:ODE0} given initial condition $\numjoblimvec(0).$ If $\resourcemax$ satisfies $\resourcemax=o(n),$ then for any fixed $T>0$, the following holds:
\begin{align}
\lim_{\numserver\to\infty}\sup_{0\le t\le T}\sum_{\type=1}^{\numtype}\left|\frac{1}{\lambdasup_\type}\numjobup_\type(t)-\numjoblim_{\type}(t)\right|=0,\quad\text{in probability.} \label{eq:finite_time_as}
\end{align}
\end{theorem}


\section{Proofs for Stability and Diminishing Queueing Probability}\label{sec:proof-diminishing}
\newcommand{\classseq}{\bm{u}}
\newcommand{\stateentryup}{U^{(\numserver)}}
\newcommand{\syssize}{x}
\newcommand{\syssizevec}{\bm{\syssize}}
\newcommand{\queue}{q}
\newcommand{\queuevec}{\bm{\queue}}
\newcommand{\nondom}{\theta}

In this section we present the proofs of stability (Theorem~\ref{thm:stability}) and diminishing queueing probability (Theorem~\ref{thm:diminishing-queueing}).  Since the structure of our proofs is similar in spirit to that of the recently-developed drift method \cite{ErySri_12,MagSri_16,WanMagSri_18}, we first provide some background on the drift method in Section~\ref{subsec:background}.  We then sketch our approach in Section~\ref{subsec:tech-diminishing} and highlight how it deviates from the traditional drift method.  Finally, we present the detailed proofs in Section~\ref{subsec:proofs-stability-diminishing}.

\subsection{Preliminaries on drift method}\label{subsec:background}
The general idea of the drift method is to construct an appropriate Lyapunov function and then study the drift of the Lyapunov function.  To be more concrete, consider the Markov chain $\stateup$ that describes the state of our multi-server job queueing system, and let $g\colon \mU \to \mathbb{R}_+$ be a Lyapunov (nonnegative) function on the state space.  Let $\classseq$ denote a state and $r_{\classseq,\classseq'}$ denote the transition rate from a state $\classseq$ to another state $\classseq'$.  Then the drift of $g$ is defined as
\begin{align*}
\Delta g(\classseq)&=\lim_{\delta\to 0}\frac{1}{\delta}\expect\left[g\left(\stateup(t+\delta)\right)-g\left(\stateup(t)\right)\Bigm|g\left(\stateup(t)\right)=\classseq\right]\\
&=\sum_{\classseq'\in\mU}r_{\classseq,\classseq'}\left(g\left(\classseq'\right)-g\left(\classseq\right)\right).
\end{align*}
From the definition, it can be seen that the drift $\Delta g$ is a function of the state $\classseq$.  When $\Delta g$ is applied to the steady state $\stateup(\infty)$, the drift $\Delta g\left(\stateup(\infty)\right)$ is referred to as the \emph{steady-state drift}.

The drift method utilizes the relationship that the expected steady-state drift is zero for well-behaved Lyapunov functions, i.e., $\expect\left[\Delta g\left(\stateup(\infty)\right)\right]=0$.
One key to the drift method is to craft an appropriate Lyapunov function $g$ such that its drift, $\Delta g\left(\stateup(\infty)\right)$, decomposes into terms that correspond to the metric of interest (e.g., the total number of jobs in the system) and terms that are tractable to bound.  There are mainly two general approaches to constructing such Lyapunov functions. The traditional approach \cite{ErySri_12,MagSri_16,WanMagSri_18} is based on establishing state-space collapse, a property whereby the state is concentrated around a strict subset of the entire state space in heavy load.  Another approach \cite{LiuYin_20,LiuGonYin_20,LiuYin_19,WenWan_21,WenZhoSri_21} is based on solving the so-called Stein's equation after coupling the system with a simple fluid model.

However, it is hard to apply either of these two approaches to our problem.
The first approach is most effective when the load scaling is in the traditional heavy-traffic regime, where the number of servers stays constant and the load approaches $1$; this is not our setting.  The second approach requires finding a reasonable fluid model that admits a solution to Stein's equation. Such a fluid representation is hard to find for our ``jittery'' model, where jobs sometimes enter service in batches.

\subsection{Our approach}\label{subsec:tech-diminishing}
Our construction of the Lyapunov function $g$ does not fall under either of the two typical approaches in the literature, although our approach follows the general framework of the drift method in the sense that we also exploit the identity $\expect\left[\Delta g\left(\stateup(\infty)\right)\right]=0$.

We consider a Lyapunov function $g\colon \mU\rightarrow \mathbb{R}_+$ defined as follows:
\begin{equation}\label{eq:g}
g(\classseq)=\sum_{\type=1}^{\numtype}\frac{\resourceup_{\type}\syssize_{\type}}{\mu_{\type}},
\end{equation}
where $\syssize_{\type}$ denotes the number of class~$\type$ jobs in system (in queue plus in service), with the understanding that $\syssize_{\type}$ is a function of the state $\classseq$.  Note that $g(\classseq)$ has the following intuitive meaning:  Since each class~$\type$ job needs to occupy $\resourceup_{\type}$ servers for an average of $1/\mu_{\type}$ duration of time, we say that the \emph{expected work} contributed by a class~$\type$ job, measured by this space-time product, is $\resourceup_{\type}/\mu_{\type}$.  Then $g(\classseq)$ represents the total amount of expected work in the system.

\sloppypar{
We then extract information from the drift $\Delta g(\classseq)$ in the following way.  We first derive the following upper-bounding function on $\Delta g(\classseq)$, which will be formally stated as Lemma~\ref{lem:g-h} in Section~\ref{subsec:proofs-stability-diminishing}:
\begin{equation*}
\Delta g(\classseq)\le \resourcemax +
\begin{cases}
\numserver\loadup-\sum_{\type=1}^{\numtype}\resourceup_{\type}\syssize_{\type}, & \text{if }\sum_{\type=1}^{\numtype}\resourceup_{\type}\syssize_{\type}\le \numserver-\resourcemax,\\
-\numserver\left(1-\loadup\right), & \text{if }\sum_{\type=1}^{\numtype}\resourceup_{\type}\syssize_{\type}>\numserver-\resourcemax.
\end{cases}
\end{equation*}
Observe that the second term in the upper-bounding function corresponds to the condition $\sum_{\type=1}^{\numtype}\resourceup_{\type}\syssize_{\type}>\numserver-\resourcemax$.   This will allow us to relate $\Pr\left(\sum_{\type=1}^{\numtype}\resourceup_{\type}\numjobup_{\type}(\infty)> \numserver-\resourcemax\right)$ to $\expect\left[\Delta g\left(\stateup(\infty)\right)\right]$, where recall that $\numjobup_{\type}(\infty)$ denotes the number of class~$\type$ jobs in steady state.  Then utilizing the identity $\expect\left[\Delta g\left(\stateup(\infty)\right)\right]=0$ will eventually lead to an upper bound on the queueing probability $\PrQup$.}

\subsection{Proofs}\label{subsec:proofs-stability-diminishing}
Consider the Lyapunov function $g$ defined in \eqref{eq:g} and let $\Delta g(\classseq)$ denote its drift.  We organize our proofs into three parts:  we first establish the upper-bounding function on $\Delta g(\classseq)$ in Lemma~\ref{lem:g-h}, which underpins all of our analysis; we then prove the stability result (Theorem~\ref{thm:stability}) based on Lemma~\ref{lem:g-h}; finally we prove the results on queueing probability (Theorem~\ref{thm:diminishing-queueing}) through several more lemmas.  The flow chart of the proofs is given in Figure~\ref{fig:proof-flow-chart}.
\begin{figure}[h!]
    \centering
    \includegraphics[scale=0.5]{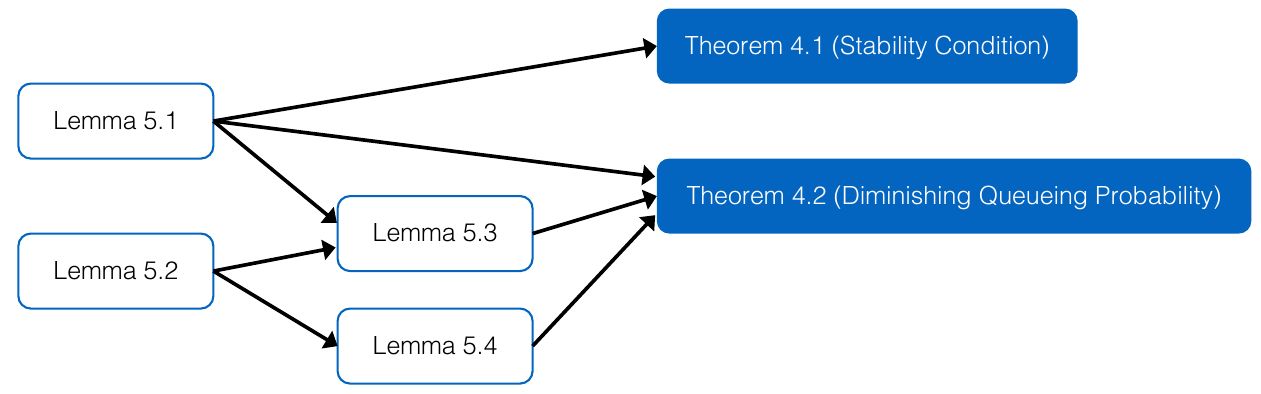}
    \caption{Flowchart of proofs}
    \label{fig:proof-flow-chart}
\end{figure}

\subsubsection*{\textbf{Upper-bounding function on the drift}}
\begin{lemma}\label{lem:g-h}
The drift of $g(\classseq)$ can be bounded as:
\begin{equation}\label{eq:g-h}
\Delta g(\classseq)\le \resourcemax + h(\syssizevec),
\end{equation}
where $\syssizevec=(\syssize_1,\syssize_2,\dots,\syssize_{\numtype})$ with $\syssize_{\type}$ denoting the number of class~$\type$ jobs in system (in queue plus in service), and the function $h$ is defined as:
\begin{equation}\label{eq:h}
h(\syssizevec)=
\begin{cases}
\numserver\loadup-\sum_{\type=1}^{\numtype}\resourceup_{\type}\syssize_{\type}, & \text{if }\sum_{\type=1}^{\numtype}\resourceup_{\type}\syssize_{\type}\le \numserver-\resourcemax,\\
-\numserver\left(1-\loadup\right), & \text{if }\sum_{\type=1}^{\numtype}\resourceup_{\type}\syssize_{\type}>\numserver-\resourcemax.
\end{cases}
\end{equation}
\end{lemma}

\begin{proof}
Recall that $g(\classseq)=\sum_{\type=1}^{\numtype}\frac{\resourceup_{\type}\syssize_{\type}}{\mu_{\type}}$.  Let $\queue_{\type}$ be the corresponding number of class~$\type$ jobs in the queue, with the understanding that $\queue_{\type}$ is a function of the state $\classseq$.
When there is a class~$\type$ arrival, which happens at a transition rate of $\lambdasup_{\type}$, the value of $\syssize_{\type}$ increases by $1$; when a class~$\type$ job departs, which happens at a transition rate of $\mu_{\type}(\syssize_{\type}-\queue_{\type})$, the value of $\syssize_{\type}$ decreases by $1$.  Therefore, the drift of $g(\classseq)$ can be written as follows:

\begin{align*}
\Delta g(\classseq)&=\sum_{\type=1}^{\numtype}\lambdasup_{\type}
\cdot\frac{\resourceup_{\type}}{\mu_{\type}}-\sum_{\type=1}^{\numtype}\mu_{\type}(\syssize_{\type}-\queue_{\type})\cdot\left(-\frac{\resourceup_{\type}}{\mu_{\type}}\right).
\end{align*}
Then noticing the relation $\loadup_{\type}=\frac{\lambdasup_{\type}\resourceup_{\type}}{\numserver\mu_{\type}}$, we get
\begin{align}
\Delta g(\classseq)
&= \numserver\loadup-\sum_{\type=1}^{\numtype}\resourceup_{\type}(\syssize_{\type}-\queue_{\type}).\label{eq:drift-3}
\end{align}
Consider the term $\sum_{\type=1}^{\numtype}\resourceup_{\type}(\syssize_{\type}-\queue_{\type})$ in \eqref{eq:drift-3}.  It is easy to see that when $\sum_{\type=1}^{\numtype}\resourceup_{\type}\syssize_{\type}\le \numserver-\resourcemax$, which corresponds to the \emph{first case} in the definition of $h$ in \eqref{eq:h},
there must be no jobs in the queue, and thus $\sum_{\type=1}^{\numtype}\resourceup_{\type}(\syssize_{\type}-\queue_{\type})=\sum_{\type=1}^{\numtype}\resourceup_{\type}\syssize_{\type}$.
When $\sum_{\type=1}^{\numtype}\resourceup_{\type}\syssize_{\type}>\numserver-\resourcemax$, which corresponds to the \emph{second case} in the definition of $h$ in \eqref{eq:h}, either there are no jobs in the queue and thus $\sum_{\type=1}^{\numtype}\resourceup_{\type}(\syssize_{\type}-\queue_{\type})=\sum_{\type=1}^{\numtype}\resourceup_{\type}\syssize_{\type}>\numserver-\resourcemax$, or the number of idle servers is not enough to absorb the job at the head of the queue and thus $\sum_{\type=1}^{\numtype}\resourceup_{\type}(\syssize_{\type}-\queue_{\type})>\numserver-\resourcemax$.  In both scenarios, $\Delta g(\classseq)\le \numserver\loadup-\numserver+\resourcemax= h(\syssizevec)+\resourcemax$, which completes the proof of Lemma~\ref{lem:g-h}.
\end{proof}

\subsubsection*{\textbf{Proof of Theorem~\ref{thm:stability} (Stability Condition)}}

We invoke the Foster-Lyapunov criteria \cite{SriYin_14} to show that the Markov chain $\stateup$ is positive recurrent.
Let $\mathcal{B}=\left\{\classseq\colon \sum_{\type=1}^{\numtype}\resourceup_{\type}\syssize_{\type}\le \numserver-\resourcemax\right\}$.  Then clearly $\mathcal{B}\subseteq \mU$ is a finite set. The drift upper bound in Lemma~\ref{lem:g-h} implies that:
\begin{itemize}[leftmargin=1.7em]
\item $\Delta g(\classseq) \le \numserver\loadup + \resourcemax < +\infty$ when $\classseq\in\mathcal{B}$;
\item $\Delta g(\classseq) \le -\numserver\left(1-\loadup\right)+\resourcemax < 0$ when $\classseq\notin \mathcal{B}$,
\end{itemize}
where the second item follows from the assumption that $\loadup < 1-\frac{\resourcemax}{\numserver}$.
Therefore, by the Foster-Lyapunov theorem, the Markov chain $\stateup$ is positive recurrent. \qed

\subsubsection*{\textbf{Proof of Theorem~\ref{thm:diminishing-queueing} (Diminishing Queueing Probability)}}
The proof of Theorem~\ref{thm:diminishing-queueing} is based on two more lemmas (Lemmas~\ref{lem:g-bound} and \ref{lem:ssc}) besides Lemma~\ref{lem:g-h}.  To keep it clear and short, we defer the proofs of Lemmas~\ref{lem:g-bound} and \ref{lem:ssc} until the end.

First we note that $\expect\left[\Delta g\left(\stateup(\infty)\right)\right]=0$ since $\expect\left[g\left(\stateup(\infty)\right)\right]<+\infty$ by Lemma~\ref{lem:g-bound}.  Then utilizing $\expect\left[\Delta g\left(\stateup(\infty)\right)\right]=0$,
we can derive an upper bound on $\Pr\left(\sum_{\type=1}^{\numtype}\resourceup_{\type}\numjobup_{\type}(\infty)> \numserver-\resourcemax\right)$, where recall that the queueing probability $\PrQup\le \Pr\left(\sum_{\type=1}^{\numtype}\resourceup_{\type}\numjobup_{\type}(\infty)> \numserver-\resourcemax\right)$. By the drift upper bound $\Delta g(\classseq)\le h(\syssizevec)+\resourcemax$ in Lemma~\ref{lem:g-h}, we have
\begin{equation*}
\expect\left[h\left(\numjobvecup(\infty)\right)\right]\ge \expect\left[\Delta g\left(\stateup(\infty)\right)\right]-\resourcemax=-\resourcemax.
\end{equation*}
Since $\left[h\left(\numjobvecup(\infty)\right)\right]$ can be written as follows, by its construction:
\begin{align*}
\expect\left[h\left(\numjobvecup(\infty)\right)\right]
&=\expect\left[\left(\numserver\loadup-\sum_{\type=1}^{\numtype}\resourceup_{\type}\numjobup_{\type}(\infty)\right)\cdot \mathds{1}_{\left\{\sum_{\type=1}^{\numtype}\resourceup_{\type}\numjobup_{\type}(\infty)\le\numserver-\resourcemax\right\}}\right]\\
&\mspace{22mu}-\numserver\left(1-\loadup\right)\Pr\left(\sum_{\type=1}^{\numtype}\resourceup_{\type}\numjobup_{\type}(\infty)>\numserver-\resourcemax\right),
\end{align*}
it follows that
\begin{align}
&\mspace{22mu}\Pr\left(\sum_{\type=1}^{\numtype}\resourceup_{\type}\numjobup_{\type}(\infty)>\numserver-\resourcemax\right)\nonumber\\
&\le \frac{\resourcemax}{\numserver\left(1-\loadup\right)}
+\frac{1}{\numserver\left(1-\loadup\right)}\expect\left[\left(\numserver\loadup-\sum_{\type=1}^{\numtype}\resourceup_{\type}\numjobup_{\type}(\infty)\right)\cdot \mathds{1}_{\left\{\sum_{\type=1}^{\numtype}\resourceup_{\type}\numjobup_{\type}(\infty)\le\numserver-\resourcemax\right\}}\right].\nonumber\\
&\le \frac{\resourcemax}{\numserver\left(1-\loadup\right)}
+\frac{1}{\numserver\left(1-\loadup\right)}\sum_{\type=1}^{\numtype}\expect\left[\left(\numserver\loadup_{\type}-\resourceup_{\type}\numjobup_{\type}(\infty)\right)^+\right],\label{eq:upper-3}
\end{align}
where $(a)^+$ for a real number $a\in\mathbb{R}$ denotes $\max\{a,0\}$.

We then bound each summand in the second term in \eqref{eq:upper-3} using Lemma~\ref{lem:ssc}, which asserts that $\expect\left[\left(\numserver\loadup_{\type}-\resourceup_{\type}\numjobup_{\type}(\infty)\right)^+\right]
\le 3\sqrt{\numserver \resourceup_{\type}}$.  Intuitively, Lemma~\ref{lem:ssc} says that the number of class~$\type$ jobs, $\numjobup_{\type}(\infty)$, cannot be much smaller than $\frac{\numserver\loadup_{\type}}{\resourceup_{\type}}=\frac{\lambdasup_{\type}}{\mu_{\type}}$.  To see this, consider the scenario when the number of class~$\type$ jobs is smaller than $\frac{\lambdasup_{\type}}{\mu_{\type}}$.  Then the departure rate of class~$\type$ jobs will definitely be much smaller than $\mu_{\type}\cdot\frac{\lambdasup_{\type}}{\mu_{\type}}=\lambdasup_{\type}$, i.e., the departure rate is smaller than the arrival rate for class~$\type$ jobs.  Consequently, the number of class~$\type$ jobs will increase.  The threshold value for the number of class~$\type$ jobs to balance the departure and arrival rates is $\frac{\lambdasup_{\type}}{\mu_{\type}}$.
Conceptually, Lemma~\ref{lem:ssc} is similar to the ``state-space collapse'' type of results in the literature, since it can be interpreted as a concentration of $\numjobup_{\type}(\infty)$ around the subset of values no smaller than $\frac{\lambdasup_{\type}}{\mu_{\type}}$.

With this upper bound given by Lemma~\ref{lem:ssc}, \eqref{eq:upper-3} can be further upper bounded as:
\begin{align}
\Pr\left(\sum_{\type=1}^{\numtype}\resourceup_{\type}\numjobup_{\type}(\infty)>\numserver-\resourcemax\right)
&\le\frac{\resourcemax}{\numserver\left(1-\loadup\right)}
+\frac{1}{\numserver\left(1-\loadup\right)}\cdot 3\numtype\sqrt{\numserver\resourcemax}\nonumber\\
&= \frac{1}{1-\loadup}\left(3\numtype\sqrt{\frac{\resourcemax}{\numserver}}+\frac{\resourcemax}{\numserver}\right).\nonumber
\end{align}
Then the diminishing queueing probability result in \eqref{eq:queueing-prob-diminish} follows immediately.  This completes the proof of Theorem~\ref{thm:diminishing-queueing}. \qed

\subsubsection*{\textbf{Lemmas~\ref{lem:g-bound} and \ref{lem:ssc} (needed in the proof of Theorem~\ref{thm:diminishing-queueing})}}

Before we present Lemmas~\ref{lem:g-bound} and \ref{lem:ssc}, we first state Lemma~\ref{lem:bounds-drift} below, which is the tool we use in the proofs of Lemmas~\ref{lem:g-bound} and \ref{lem:ssc} (see the flowchart in Figure~\ref{fig:proof-flow-chart}).  Lemma~\ref{lem:bounds-drift} is a well-known result that bounds tail probabilities and moments using drift conditions \cite{Haj_82,BerGamTsi_01,WanMagSri_18}.  We include it here for completeness, and the form below slightly generalizes the commonly used form in the literature.
\begin{lemma}[Bounds via Drift]\label{lem:bounds-drift}
Let $\{S(t),t\ge 0\}$ be a continuous-time Markov chain on a countable state space $\mathcal{S}$ and $r_{s,s'}$ be its transition rate from state $s$ to state $s'$.  Assume that it has a unique stationary distribution and let $S(\infty)$ be a random element that follows the stationary distribution.  Let $V\colon \mathcal{S}\rightarrow\mathbb{R}_+$ be a Lyapunov function. Suppose
\begin{equation*}
v_{\max}:=\sup_{s,s'\colon r_{s,s'}>0}|V(s')-V(s)|<+\infty,\qquad \sup_{s}\sum_{s'}r_{s,s'}<+\infty,
\end{equation*}
and let
\begin{equation*}
\delta_{\max}=\sup_{s}\sum_{s'\colon V(s')>V(s)}r_{s,s'}|V(s')-V(s)|.
\end{equation*}
Suppose that there exist $B>0$ and $\gamma>0$ such that for any $s$ with $V(s)>B$,
\begin{equation*}
\Delta V(s)\le -\gamma.
\end{equation*}
Then for all nonnegative integers $m$,
\begin{equation*}
\Pr\Big(V\big(S(\infty)\big)>B+2mv_{\max}\Big)\le\left(\frac{\delta_{max}}{\delta_{max}+\gamma}\right)^{m+1}.
\end{equation*}
Further,
\begin{equation*}
\expect\Big[V\big(S(\infty)\big)\Big]\le B+\frac{2v_{\max}\delta_{\max}}{\gamma}.
\end{equation*}
\end{lemma}

Now we are ready for Lemmas~\ref{lem:g-bound} and \ref{lem:ssc} and their proofs.
\begin{lemma}\label{lem:g-bound}
The Lyapunov function $g$ defined in \eqref{eq:g} satisfies
\begin{equation*}
\expect\left[g\left(\stateup(\infty)\right)\right] < +\infty.
\end{equation*}
\end{lemma}
\begin{proof}
We prove Lemma~\ref{lem:g-bound} by applying Lemma~\ref{lem:bounds-drift} to the Lyapunov function $g$. Recall that $g(\classseq)=\sum_{\type=1}^{\numtype}\frac{\resourceup_{\type}\syssize_{\type}}{\mu_{\type}}$, where we use $\classseq$ to denote a state and $\syssizevec=(\syssize_1,\syssize_2,\dots,\syssize_{\numtype})$ to denote the corresponding numbers of jobs in the system of each class.

Let $B=\frac{\numserver-\resourcemax}{\mu_{\min}}$.  We bound the drift $\Delta g(\classseq)$ when $g(\classseq)>B$ using Lemma~\ref{lem:g-h}.
When $g(\classseq)>B$, we have
\begin{align*}
\sum_{\type=1}^{\numtype}\resourceup_{\type}\syssize_{\type}  \geq \mu_{\min}\cdot\sum_{\type=1}^{\numtype}\frac{\resourceup_{\type}\syssize_{\type}}{\mu_{\type}} =\mu_{\min}\cdot g(\classseq)> \numserver-\resourcemax.
\end{align*}
Then Lemma~\ref{lem:g-h} implies that $\Delta g(\classseq) \le -\numserver\left(1-\loadup\right)+\resourcemax$, which is negative when 
$\loadup < 1-\frac{\resourcemax}{\numserver}$.

Note that when we take the Lyapunov function in Lemma~\ref{lem:bounds-drift} to be $g$, then $v_{\max}=\max_{\type} \left\{\frac{\resourceup_\type}{\mu_\type}\right\}\leq \frac{\resourcemax}{\mu_{\min}}$, and $\delta_{\max}=\sum_{\type =1}^{\numtype}\frac{\resourceup_{\type}\lambdasup_{\type}}{\mu_{\type}}=\numserver\loadup$.  Applying Lemma~\ref{lem:bounds-drift}, we get
\begin{align*}
\expect \left[g\left(\stateup(\infty)\right)\right] 
\le B+ \frac{2v_{\max}\delta_{\max}}{\numserver\left(1-\loadup\right)-\resourcemax}
\le \frac{\numserver-\resourcemax}{\mu_{\min}}+2\cdot\frac{\resourcemax}{\mu_{\min}}\cdot \frac{\numserver\loadup}{\numserver\left(1-\loadup\right)-\resourcemax},
\end{align*}
which is finite, and thus completes the proof of Lemma~\ref{lem:g-bound}.
\end{proof}

\begin{lemma}\label{lem:ssc}
For each class $\type$,
\begin{equation}\label{eq:ssc}
\expect\left[\left(\numserver\loadup_{\type}-\resourceup_{\type}\numjobup_{\type}(\infty)\right)^+\right]
\le 3\sqrt{\numserver \resourceup_{\type}}.
\end{equation}
\end{lemma}
\begin{proof}
We prove Lemma~\ref{lem:ssc} by applying Lemma~\ref{lem:bounds-drift} to the Lyapunov function $f\colon \mU\rightarrow\mathbb{R}_+$ defined as
\begin{equation*}
f(\classseq)=\left(\numserver\loadup_{\type}-\resourceup_{\type}\syssize_{\type}\right)^+,
\end{equation*}
where recall that we use $\classseq$ to denote a state and $\syssizevec=(\syssize_1,\syssize_2,\dots,\syssize_{\numtype})$ to denote the corresponding numbers of jobs in the system of each class.  Then \eqref{eq:ssc} is equivalent to $\expect\left[f\left(\stateup(\infty)\right)\right]\le 3\sqrt{\numserver \resourceup_{\type}}$.

Let $B=\sqrt{\numserver\resourceup_{\type}}$.  Then we bound the drift $\Delta f(\classseq)$ when $f(\classseq)>B$. The condition that $f(\classseq)>B$ guarantees that $f(\classseq)=\numserver\loadup_{\type}-\resourceup_{\type}\syssize_{\type}>B\ge \resourceup_{\type}$. Therefore, the drift $\Delta f(\classseq)$ can be bounded as follows:
\begin{align*}
\Delta f(\classseq)
& = \lambdasup_{\type}\cdot\left(-\resourceup_{\type}\right) +\mu_{\type}(\syssize_{\type}-\queue_{\type})\cdot \resourceup_{\type}\\
&\leq-\mu_{\type}\left(\numserver\loadup_{\type}-\resourceup_{\type}\syssize_{\type}\right)\\
&< -\mu_{\type} B,
\end{align*}
where the last line follows from the condition that $f(\classseq)>B$. Thus, the $\gamma$ in Lemma~\ref{lem:bounds-drift} equals $\mu_{\type}B$.

Note that when we take the Lyapunov function in Lemma~\ref{lem:bounds-drift} to be $f$, then $v_{\max}=\resourceup_{\type}$, and $\delta_{\max}$
satisfies $\delta_{\max}\le \frac{\numserver\mu_{\type}}{\resourceup_{\type}}\cdot \resourceup_{\type}=\numserver\mu_{\type}$, since $f$ increases when there is a departure of class~$\type$ jobs.
Applying Lemma~\ref{lem:bounds-drift}, we get 
\begin{align*}
\expect\left[f\left(\stateup(\infty)\right)\right]
&\le B+\frac{2\numserver\resourceup_{\type}\mu_{\type}}{\mu_{\type}B} = 3\sqrt{\numserver\resourceup_{\type}},
\end{align*}
which proves the bound \eqref{eq:ssc} in Lemma~\ref{lem:ssc}.
\end{proof}

\section{Proof for Transient Analysis}\label{sec:proof-limiting-distr}
In this section, we focus on analyzing the transient behavior of the number of jobs of each class (Theorem~\ref{THM:FINITE_TIME}). 
Our analysis is motivated by the approach of \emph{fluid approximation}, where an appropriately scaled system process can be approximated, as the number of servers grows large, by a deterministic system that is defined by a system of differential equations. 
Such a deterministic system is referred to as \emph{fluid model} in the literature. 
In particular, our proof is closely related to the argument for a result called Kurtz's theorem on density-dependent Markov processes (see Chapter $8$ of \cite{Kur_81} or Chapter $5.3$ of \cite{DraMas_09}). However, we remark that there are two key differences between our proof and the standard argument. First, we scale the number of jobs of each class by the corresponding arrival rate; in contrast, traditional fluid approximation considers scaling by the number of servers. The traditional fluid scaling does not work for a system with multi-server jobs, where server needs can potentially grow with the number of servers $n$.
Second, because we have multiple job classes, instead of directly comparing our original system with the fluid model, we need to construct an intermediate system and couple it with our original system. We establish the fluid approximation by showing that the intermediate system is close to both the original system and the fluid model.

We first introduce the fluid model and its properties in Section~\ref{subsec:ODE}, and then prove Theorem~\ref{THM:FINITE_TIME} in Section~\ref{subsec:finite_time}. 
For ease of exposition, we denote the scaled number of class $\type$ jobs in the system at time $t$ by 
\begin{align*}
    \scalenumjobup_\type(t):=\frac{\numjobup_\type(t)}{\lambdasup_\type} .
\end{align*}

\subsection{Fluid model} \label{subsec:ODE}

Recall that the fluid model is defined by the following differential equations:
\begin{equation}\label{eq:ODE}
\dot{\numjoblim}_\type(t) = 1 - \min\left\{\mu_\type\numjoblim_\type(t),\frac{1}{\load_{\type}}\right\},\qquad \type\in\{1,\ldots,\numtype\}.
\end{equation}
For each $\type\in \{1,\ldots,\numtype\}$, given an initial condition 
satisfying $\numjoblim_\type(0)\le \frac{1}{\mu_\type\load_\type}$, it is not hard to see that the unique solution to the differential equation~\eqref{eq:ODE} is given by
\begin{align}\label{eq:solution_ODE_1}
\numjoblim_\type(t)=\left(\numjoblim_\type(0)-\frac{1}{\mu_\type}\right)e^{-\mu_\type t}+\frac{1}{\mu_\type}.
\end{align}


Observe that starting from any initial condition $\numjoblimvec(0)$ satisfying $0\leq \numjoblim_\type(0)\leq  \frac{1}{\mu_\type\load_\type}$ for all $\type$, as $t\rightarrow \infty,$ the system $\numjoblimvec(t)$ converges to the following equilibrium point 
$$\numjoblimvec(\infty):=\Big(\frac{1}{\mu_1}, \frac{1}{\mu_2},\ldots,\frac{1}{\mu_\numtype}\Big).$$

\emph{Interpretation of the differential equations.} Consider a system where each class of jobs has access to a separate set of $\numserver$ servers and is served in FCFS order, i.e., class $\type$ jobs run as an M/M/$\frac{\numserver}{\resourceup_\type}$ queue with arrival rate $\lambdasup_\type$ and service rate $\mu_\type$. For simplicity, assume that $n/\resourceup_\type$ is an integer. 
Let $Y^{(n)}_\type(t)$ denote the number of jobs in this system at time $t$. 

We can view a solution to the fluid model~(\ref{eq:ODE}), $\numjoblim_\type(t)$, as a deterministic approximation to the sample paths of $\frac{Y^{(n)}_\type(t)}{\lambdasup_\type}$ for a large $n.$
Note that  $\frac{Y^{(n)}_\type(t)}{\lambdasup_\type}$ decreases by $\frac{1}{\lambdasup_\type}$ at rate $\min\big\{\mu_\type Y^{(n)}_\type(t), \mu_\type\frac{n}{\resourceup_\type} \big\}$, due to the completion of a job. Multiplying the departure rate by the decrease due to a departure, and taking the limit $n\rightarrow\infty,$  we obtain the second drift term $\min\big\{\mu_\type\numjoblim_\type(t),\frac{1}{\load_{\type}}\big\}$ in (\ref{eq:ODE}). The first drift term, $1$, corresponds to arrivals, as $\frac{Y^{(n)}_\type(t)}{\lambdasup_\type}$ increases by $\frac{1}{\lambdasup_\type}$ at rate $\lambdasup_\type.$

\subsection{Proof outline of Theorem~\ref{THM:FINITE_TIME}} \label{subsec:finite_time}

Here we provide a proof outline of Theorem~\ref{THM:FINITE_TIME} and highlight the difference from the standard argument for Kurtz's theorem. We defer the detailed proof to Appendix~\ref{proof_finite_time}. 

As we mentioned earlier, a key difference between our proof and the standard argument is the construction of an intermediate system. 
In particular, we consider the system $\{\numjobavecup(t)\}$ constructed in Section~\ref{subsec:ODE}, i.e., $\numjobaup_{\type}(t)$ is the number of jobs in an M/M/$\frac{\numserver}{\resourceup_\type}$ queue at time $t$. We couple this enlarged
system with our system so that they have the same initial state, identical
job arrival sequence and identical job service times. With this coupling,
the system $\{\numjobavecup(t)\}$ is identical to our system (in terms of the number of
present jobs of each class) until the moment when the total number
of servers requested by present jobs exceeds $\numserver$. Specifically,
for any positive time $T$,
we have $\numjobup_{\type}(t)=\numjobaup_{\type}(t)$ for all classes
$\type$ and all $t$ with $0\le t\le T$ if 
\begin{align*}
\sup_{0\le t\le T}\sum_{\type=1}^{\numtype}\resourceup_{\type}\numjobaup_{\type}(t)\le\numserver.
\end{align*}

To establish that the scaled number of jobs in our original system, $\Big\{\frac{1}{\lambdasup_{\type}}\numjobup_{\type}(t)\Big\}$, can be approximated by the fluid model $\numjoblimvec(t)$, it suffices to show that the scaled number of jobs in the intermediate enlarged system,  $\Big\{\frac{1}{\lambdasup_{\type}}\numjobaup_{\type}(t)\Big\}$, is close to both of $\Big\{\frac{1}{\lambdasup_{\type}}\numjobup_{\type}(t)\Big\}$ and the fluid model $\numjoblimvec(t)$. Specifically, note that 
\begin{align}
\Pr\left(\sup_{0\le t\le T}\sum_{\type=1}^{\numtype}\left|\frac{1}{\lambdasup_{\type}}\numjobup_{\type}(t)-\numjoblim_{\type}(t)\right|>\epsilon\right)  & \leq
\sum_{\type=1}^{\numtype}\Pr\left(\sup_{0\le t\le T}\left|\frac{1}{\lambdasup_{\type}}\numjobaup_{\type}(t)-\numjoblim_{\type}(t)\right|>\frac{\epsilon}{\numtype}\right) \label{eq:term-aux} \\
& +\Pr\left(\exists \text{ \ensuremath{\type} and \ensuremath{t} $\in [0,T]$ s.t. } \numjobup_{\type}(t)\neq\numjobaup_{\type}(t)\right)\label{eq:deviate}.
\end{align}

Our proof consists of the following three main steps:
\begin{itemize}[leftmargin=5em]
\item[\textbf{Step 1:}] We show that $\Big\{\frac{1}{\lambdasup_{\type}}\numjobaup_{\type}(t)\Big\}$ is close to the fluid model $\numjoblimvec(t)$. In particular, we upper bound \eqref{eq:term-aux} 
by showing that for each $\type$,
\begin{align}
 & \Pr\left(\sup_{0\le t\le T}\left|\frac{1}{\lambdasup_{\type}}\numjobaup_{\type}(t)-\numjoblim_{\type}(t)\right|>\frac{\epsilon}{\numtype}\right)\nonumber \\
\leq & \Pr\bigg(\left|\frac{1}{\lambdasup_{\type}}\numjobaup_{\type}(0)-\numjoblim_{\type}(0)\right|>\frac{\epsilon e^{-\mu_{\type}T}}{3\numtype} \bigg)+ 2e^{-\lambdasup_{\type}T\cdot h\left(\frac{\epsilon e^{-\mu_{\type}T}}{3\numtype T}\right)}+2e^{-\frac{\numserver\mu_{\type}T}{\resourceup_{\type}}h\left(\frac{\load_{\type}\epsilon e^{-\mu_{\type}T}}{3\numtype T}\right)},\label{eq:bound-aux-n}
\end{align}
using ideas similar to those used in proving Kurtz's theorem~\cite{Kur_81,DraMas_09}.
\item[\textbf{Step 2:}] We next show that $\numjobaup_{\type}(t)$ is close to $\numjobup_{\type}(t)$. Specifically, we upper bound the probability of $\numjobup_{\type}(t)$ deviating from $\numjobaup_{\type}(t)$ in \eqref{eq:deviate} utilizing
the result from Step~1.

\item[\textbf{Step 3:}] We combine the results from Step 1-2 to show that $$\lim_{n\rightarrow \infty}\Pr\left(\sup_{0\le t\le T}\sum_{\type=1}^{\numtype}\left|\frac{1}{\lambdasup_{\type}}\numjobup_{\type}(t)-\numjoblim_{\type}(t)\right|>\epsilon\right)=0,$$ 
thus completing the proof of Theorem~\ref{THM:FINITE_TIME}.
\end{itemize}

\begin{remark} Theorem~\ref{THM:FINITE_TIME}
and the convergence of $\numjoblimvec(t)$ to $\numjoblimvec(\infty)$ give the following result:
\begin{align} \label{convergence_time_n}
    \scalenumjobup(t) \stackrel{d}{\rightarrow} \numjoblimvec(t),\text{ as }n\rightarrow \infty,\qquad\numjoblimvec(t)\rightarrow  \numjoblimvec(\infty),\text{ as }t\rightarrow \infty,
\end{align}
where ``$\stackrel{d}{\rightarrow}$'' denotes convergence in distribution. 
On the other hand, when the system load $\load$ satisfies $\load<1-\frac{\resourcemax}{\numserver}$,
Theorem~\ref{thm:stability} implies that $\scalenumjobvecup$ has a unique stationary distribution. That is, for any $\numserver$,
\begin{align*}
\scalenumjobup(t)\stackrel{d}{\rightarrow} \scalenumjobvecup(\infty),\text{ as }t\rightarrow \infty,
\end{align*}
where $\scalenumjobvecup(\infty)$ is the random variable with the stationary distribution of $\scalenumjobvecup$. 
We may expect that equation~(\ref{convergence_time_n}) still holds if we change the order in which the limits over $t$ and $\numserver$ are taken, i.e.,
\begin{align} \label{convergence_n_time}
   \scalenumjobup(t)\stackrel{d}{\rightarrow} \scalenumjobvecup(\infty),\text{ as }t\rightarrow \infty,\qquad
   \scalenumjobvecup(\infty)\stackrel{d}{\rightarrow}
   \numjoblimvec(\infty),\text{ as }n\rightarrow \infty.
\end{align}
Unfortunately, existing techniques fall short of establishing such an \emph{interchange of limits}. However, we conjecture that  $\{\scalenumjobvecup(\infty)\}_{\numserver}$ converges to  $\numjoblimvec(\infty)$ as $\numserver \rightarrow \infty.$ In particular, our numerical experiments appear to support our conjecture (see Section~\ref{sec:simulations}). We leave as an intriguing open question establishing the convergence of $\{\scalenumjobvecup(\infty)\}_{\numserver}$.


\end{remark}

\section{Simulation Results}\label{sec:simulations}
\newcommand{\figscale}{0.38}

\begin{figure}[b]
\begin{minipage}[t]{.45\textwidth}
\centering
\includegraphics[scale=\figscale]{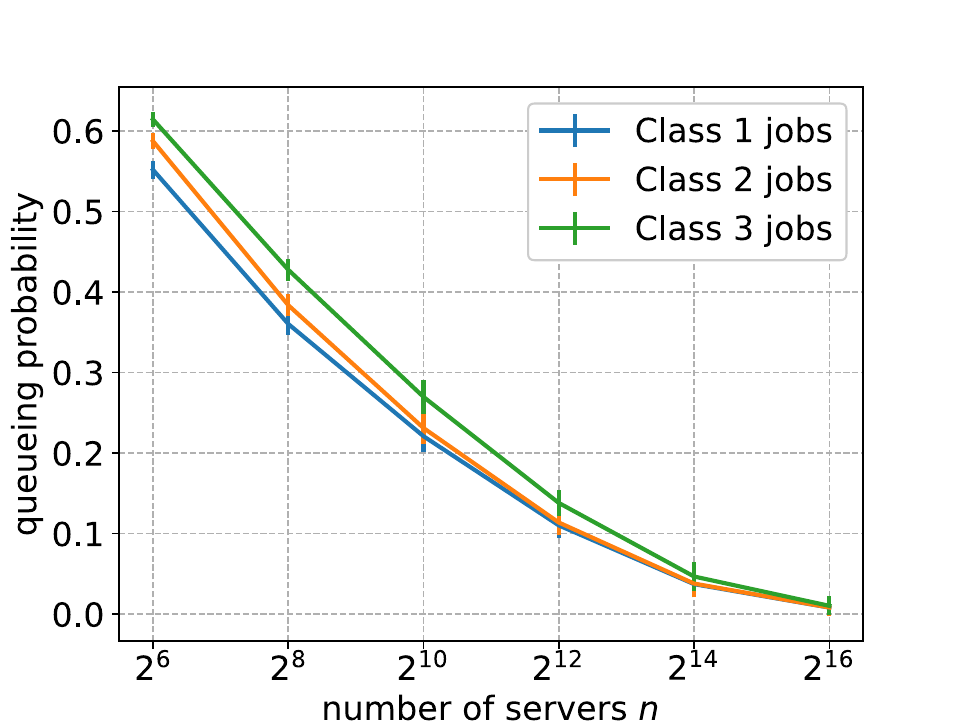}
\caption{Diminishing queueing probabilities. System load: $\loadup=1-\frac{1}{4}\numserver^{-0.1}$; server needs: $\left(\resourceup_1,\resourceup_2,\resourceup_3\right)=\left(3,\log_2\numserver,\sqrt{\numserver}\right)$.}
\label{fig:queueing-prob}
\end{minipage}
\hspace{0.2in}
\begin{minipage}[t]{.45\textwidth}
\centering
\includegraphics[scale=\figscale]{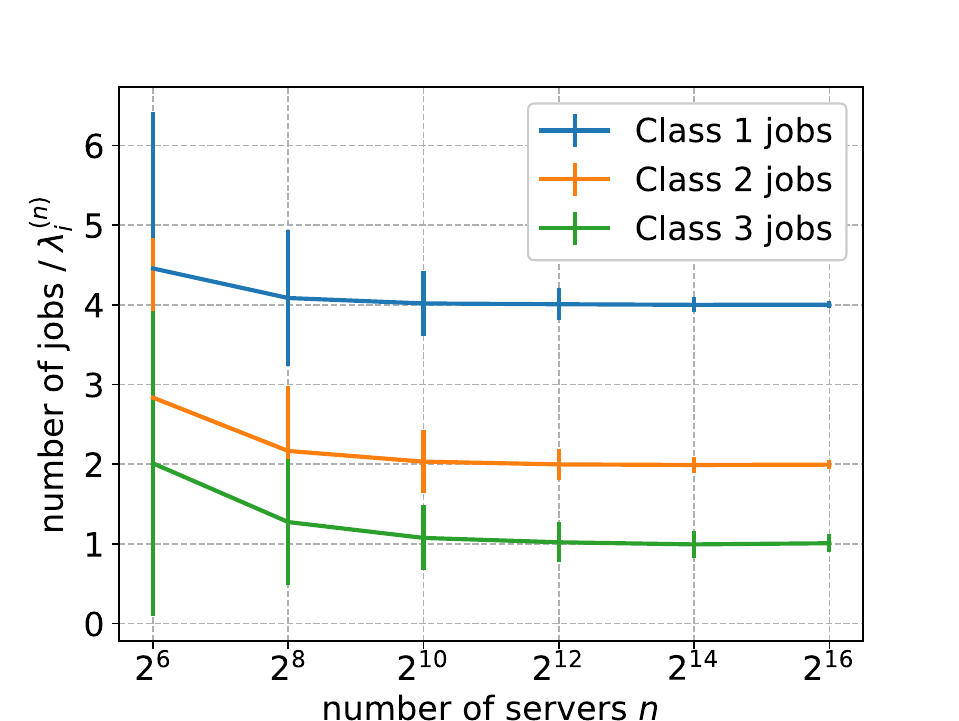}
\caption{Convergence of scaled numbers of jobs to $\frac{1}{\mu_i}$, where $\mu_1=0.25,\mu_2=0.5$, and $\mu_3=1$. System load: $\loadup=1-\frac{1}{4}\numserver^{-0.1}$; server needs: $\left(\resourceup_1,\resourceup_2,\resourceup_3\right)=\left(3,\log_2\numserver,\sqrt{\numserver}\right)$.}
\label{fig:queue-len}
\end{minipage}
\vspace{-0.2in}
\end{figure}

In this section, we perform three sets of simulation experiments to demonstrate our theoretical results and investigate gaps in the theory. 

In all the experiments, we simulate a sequence of systems with $\numserver=2^6, 2^8, 2^{10}, 2^{12}, 2^{14}, 2^{16}$ servers.  Since the scaling of jobs' server needs is a distinctive feature of our model, we first focus on a setting with three types of jobs whose server needs are $\resourceup_1=3$, $\resourceup_2=\log_2\numserver$ and $\resourceup_3=\sqrt{\numserver}$ for sets I and II; then we vary the maximum server need in set III.  The service times are exponentially distributed with rates $\mu_1=0.25, \mu_2=0.5$ and $\mu_3=1$.  Values of parameters for sets I and II are summarized in Table~\ref{tab:simulation-parameters}.

\vspace{-0.075in}
\begin{table}[!hbp]
    \small
    \centering
    \begin{tabular}{c r r r r r r}
    \toprule
    $\numserver$ & $2^6=64$ & $2^8=256$ & $2^{10}=1024$ & $2^{12}=4096$ & $2^{14}=16384$ & $2^{16}=65536$ \\
    \midrule
    I, II: $\left(\resourceup_1,\resourceup_2,\resourceup_3\right)$ & $(3,6,8)$ & $(3,8,16)$ & $(3,10,32)$ & $(3,12,64)$ & $(3, 14, 128)$ & $(3, 16, 256)$\\
    I: $\loadup=1-\frac{1}{4}\numserver^{-0.1}$ & $0.8351$ & $0.8564$ & $0.8750$ & $0.8912$ & $0.9053$ & $0.9175$\\
    II: $\loadup=1-\numserver^{-0.3}$ & $0.7128$ & $0.8105$ & $0.8750$ &      $0.9175$ & $0.9456$ & $0.9641$\\
    \bottomrule
    \end{tabular}
    \caption{Simulation parameters}
    \label{tab:simulation-parameters}
    \vspace{-0.15in}
\end{table}

\subsubsection*{\textbf{Set I}}
In the first set of experiments, our goal is to demonstrate the diminishing queueing probability and to investigate the distribution of the number of jobs from each class.  Recall that the scaling regimes where we prove a diminishing queueing probability are the four regions in Figure~\ref{fig:scaling-regimes}.  Here we pick the most interesting region, Region 4, where both the server needs and the system load $\loadup$ scale with $\numserver$.  Specifically, we choose the load to be $\loadup=1-\frac{1}{4}\numserver^{-0.1}$.  One can verify that this setting satisfies the condition for a diminishing queueing probability established by Theorem~\ref{thm:diminishing-queueing}.

\paragraph{Queueing probability} Figure~\ref{fig:queueing-prob} shows the fraction of jobs that have to queue upon arrival for each job class, which serves as an estimate of the queueing probability.  To see how reliable these estimates are, we take the last $100,0000$ data points (seen by the Poisson arrivals) and divide them into $10$ segments.  We then calculate the fraction of jobs that queue for each segment.  The points on the curves are the mean fractions averaged over the $10$ segments, and the error bars mark one standard deviation.  The small error bars in Figure~\ref{fig:queueing-prob} indicate that the estimated queueing probabilities are very stable.  The trends of the curves in Figure~\ref{fig:queueing-prob} demonstrate that the queueing probabilities are diminishing as $\numserver$ increases, as predicted by our theoretical results.

\paragraph{Number of jobs in system}
Although we were not able to theoretically analyze the number of jobs from each class in steady state, we conjecture that the number converges in distribution based on the transient analysis.  Specifically, recall that $\numjobup_{\type}(\infty)$ denotes the number of class~$\type$ jobs (in queue plus in service) in steady state.  Then we conjecture that $\frac{1}{\lambdaup_{\type}}\numjobup_{\type}(\infty)\to \frac{1}{\mu_{\type}}$ in distribution.  Our simulation results in Figure~\ref{fig:queue-len} support this conjecture.  In Figure~\ref{fig:queue-len}, the points on the curves are the average of $\frac{1}{\lambdaup_{\type}}\numjobup_{\type}(t)$ seen by Poisson arrivals when the system is empirically steady, and the error bars mark one empirical standard deviation.  Recall that $\mu_{1}=0.25, \mu_2=0.5$ and $\mu_{3} = 1$.
We can see that the curves converge nicely to the point mass distributions at the $\frac{1}{\mu_{\type}}$'s.
\begin{figure}
\begin{minipage}[t]{.45\textwidth}
\centering
\includegraphics[scale=\figscale]{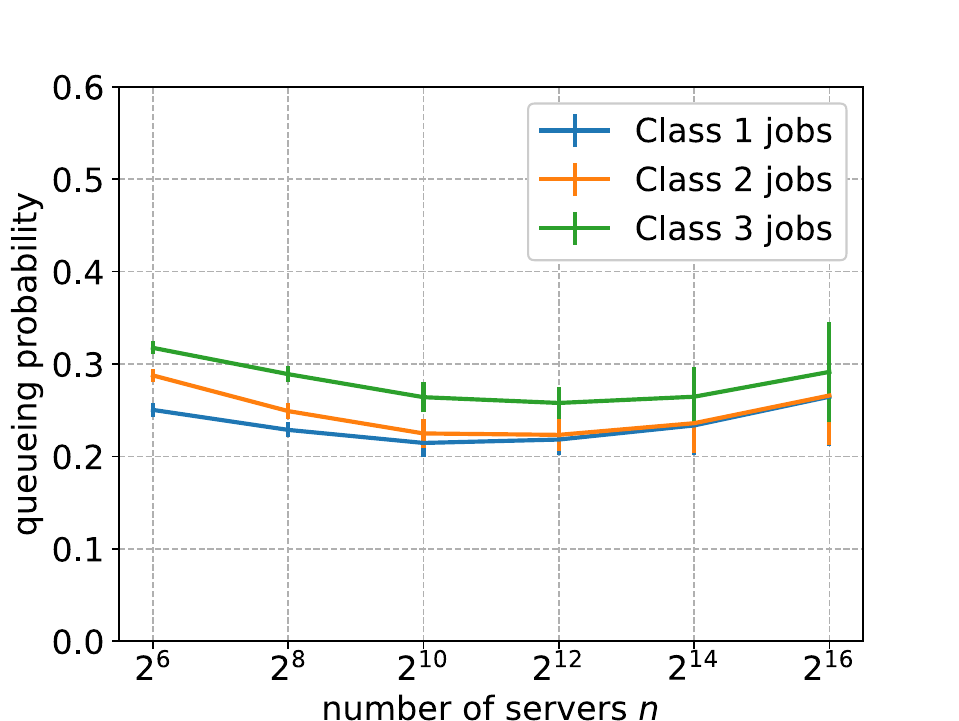}
\caption{Non-diminishing queueing probabilities. System load: $\loadup=1-\numserver^{-0.3}$; server needs: $\left(\resourceup_1,\resourceup_2,\resourceup_3\right)=\left(3,\log_2\numserver,\sqrt{\numserver}\right)$.}
\label{fig:queueing-prob-nonzero}
\end{minipage}
\hspace{0.2in}
\begin{minipage}[t]{.45\textwidth}
\centering
\includegraphics[scale=\figscale]{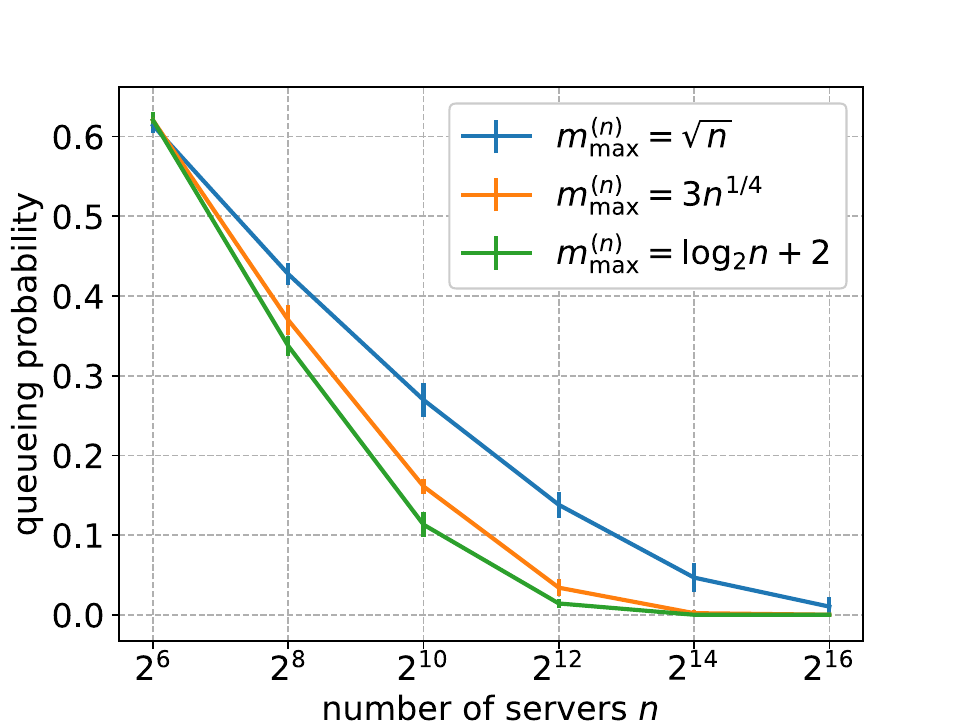}
\caption{Impact of maximum server need~$\resourcemax$. System load: $\loadup=1-\frac{1}{4}\numserver^{-0.1}$; server needs: $\left(\resourceup_1,\resourceup_2,\resourceup_3\right)=\left(3,\log_2\numserver,\sqrt{\numserver}\right)$, $\left(3,\log_2\numserver,3\numserver^{1/4}\right)$, $(3,\log_2\numserver,\log_2\numserver+2)$.}
\label{fig:queueing-prob-comp}
\end{minipage}
\vspace{-0.2in}
\end{figure}

\subsubsection*{\textbf{Set II}}
In the second set of experiments, our goal is to empirically investigate {\em when} the queueing probabilities are not diminishing.  We keep the server needs the same as those in the first set of experiments, but we increase the load to be $\loadup=1-\numserver^{-0.3}$.  This scaling regime no longer satisfies our condition for a diminishing queueing probability in Theorem~\ref{thm:diminishing-queueing}.  Figure~\ref{fig:queueing-prob-nonzero} demonstrates that indeed, it is unlikely that the queueing probabilities will converge to zero here.

\subsubsection*{\textbf{Set III}}
In the last set of experiments, we investigate how the maximum server need, $\resourcemax$, affects the queueing probability.  Figure~\ref{fig:queueing-prob-comp} compares the queueing probabilities of class~$3$ jobs under three settings where $\resourcemax$ varies as $\sqrt{\numserver}$, $3\numserver^{1/4}$, and $\log_2\numserver+2$, while keeping other parameters the same.
It shows that the queueing probability diminishes faster as $\resourcemax$ decreases.

\section{Conclusion and Future Work}\label{sec:conclusion-future-work}
In this paper, we consider a model that we refer to as the multi-server job queueing model.  In this model, a job requests multiple servers and holds on these servers simultaneously during its service.  We investigate a novel scaling regime where both the server needs of jobs and the system load scale with the total number of servers in the system.  Our main result is an upper bound on the queueing probability under FCFS scheduling, which allows us to establish conditions for the queueing probability to go to zero as the number of servers grows.  We also characterize the transient behavior of the system under constant load in the large-system limit.

There are many interesting directions that are worth further investigation in future work.  Here we list a few.
\begin{enumerate}[leftmargin=1.5em,label=(\arabic*)]
\item As we mentioned in Sections~\ref{sec:introduction} and \ref{sec:related-work}, finding exact (non-asymptotic) stability conditions under FCFS scheduling for more than two job classes with heterogeneous service rates is open.
\item Characterizing the job response time in non-asymptotic regimes is wide open.
\item It is of great interest to extend our stability condition and queueing probability upper bound to settings with service time distributions beyond exponential distributions.  For this direction, a natural first attempt would be to generalize the Lyapunov drift-based analysis in this paper.  A recent work \cite{WanMagSri_18} has demonstrated that the drift method can be used to obtain tight bounds in heavy traffic under phase-type service time distributions.  However, the analysis in \cite{WanMagSri_18} requires a complicated construction of the Lyapunov function. We anticipate that finding a proper Lyapunov function for the multi-server job model will be highly challenging.
\item Network structure is playing an increasingly important role in job scheduling due to data locality \cite{WanZhuYin_13,XieLu_15}, which constrains the servers on which a job can run.  Recent works \cite{MukBorvan_18,WenZhoSri_21,RutMuk_21} have studied how the network structure affects performance in load-balancing systems.
Analyzing the multi-server job model with a network structure is an interesting future direction.  Here the network structure can be modeled as follows: a multi-server job not only specifies how many servers it needs, but also which servers it can run on.  Then a fundamental question is: What are the conditions on job server needs, system load, and additionally, \emph{network topology}, that result in zero queueing?
\end{enumerate}

\section*{Acknowledgments.}
\begin{sloppypar}
We thank Sem Borst for his insightful comments on the paper.
This work was supported in part by NSF grants CIF-1409106, CMMI-1938909, XPS-1629444, CSR-1763701, CNS-2007733 and CNS-1955997; and by a Google 2020 Faculty Research Award.
\end{sloppypar}


\appendix
\section{\textbf{Proof of Theorem~\ref{THM:FINITE_TIME}}}
\label{proof_finite_time}

We first prove \textbf{Step 1}. Consider the queue length process $\{\numjobaup_{\type}(t)\colon t\ge0\}$
for class $\type$. It is clear that the queue length increases by
$1$ with rate $\lambdasup_{\type}$ and decreases by $1$ with rate
$\mu_{\type}\cdot \min\Big\{ \numjobaup_{\type}(t),\frac{\numserver}{\resourceup_{\type}}\Big\} $.
Let $\{N_{\type,1}(t)\colon t\ge0\}$ and $\{N_{\type,2}(t)\colon t\ge0\}$
for $\type=1,2,\dots,\numtype$ be independent unit-rate Poisson processes.
Then $\numjobaup_{\type}(t)$ can be constructed as follows:
\[
\numjobaup_{\type}(t)=\numjobaup_{\type}(0)+N_{\type,1}\left(\int_{0}^{t}\lambdasup_{\type}ds\right)-N_{\type,2}\left(\int_{0}^{t}\mu_{\type}\min\left\{ \numjobaup_{\type}(s),\frac{\numserver}{\resourceup_{\type}}\right\} ds\right).
\]

We consider a scaled version of $\numjobaup_{\type}(t)$, defined
as $\numjobaupscaled_{\type}=\frac{1}{\lambdasup_{\type}}\numjobaup_{\type}(t)$.
We have
\begin{align}
 \numjobaupscaled_{\type}(t)\nonumber 
= & \numjobaupscaled_{\type}(0)+\frac{1}{\lambdasup_{\type}}N_{\type,1}\left(\int_{0}^{t}\lambdasup_{\type}ds\right)-\frac{1}{\lambdasup_{\type}}N_{\type,2}\left(\int_{0}^{t}\lambdasup_{\type}\mu_{\type}\min\left\{ \numjobaupscaled_{\type}(s),\frac{\numserver}{\lambdasup_{\type}\resourceup_{\type}}\right\} ds\right)\nonumber \\
= & \numjobaupscaled_{\type}(0)+\frac{1}{\lambdasup_{\type}}\overline{N}_{\type,1}\left(\int_{0}^{t}\lambdasup_{\type}ds\right)-\frac{1}{\lambdasup_{\type}}\overline{N}_{\type,2}\left(\int_{0}^{t}\lambdasup_{\type}\mu_{\type}\min\left\{ \numjobaupscaled_{\type}(s),\frac{\numserver}{\lambdasup_{\type}\resourceup_{\type}}\right\} ds\right)\nonumber \\
 & +\int_{0}^{t}\left(1-\mu_{\type}\min\left\{ \numjobaupscaled_{\type}(s),\frac{\numserver}{\lambdasup_{\type}\resourceup_{\type}}\right\} \right)ds,\label{eq:mean}
\end{align}
where $\overline{N}_{\type,1}(t)=N_{\type,1}(t)-t$ and $\overline{N}_{\type,2}(t)=N_{\type,2}(t)-t$
are the centered Poisson processes.
Consider the term \eqref{eq:mean} and note that $\frac{\numserver\mu_{\type}}{\lambdasup_{\type}\resourceup_{\type}}=\frac{1}{\load_{\type}}$.
Define a function $F_{\type}\colon\mathbb{R}_{+}\rightarrow\mathbb{R}$
by $F_{\type}(x)=1-\min\left\{ \mu_{\type}x,1/\load_{\type}\right\} $.
Then (\ref{eq:mean}) can be written as $\int_{0}^{t}F_{\type}\Big(\numjobaupscaled_{\type}(s)\Big)ds$.
It can be easily verified that $F_{\type}(\cdot)$
is $\mu_{\type}$-Lipschitz. Note that the fluid model $\numjoblim_{\type}(t)$ can be
written in an integral form as follows: 
\[
\numjoblim_{\type}(t)=\numjoblim_{\type}(0)+\int_{0}^{t}\left(1-\mu_{\type}\min\left\{ \mu_{i}\numjoblim_{\type}(t),\frac{1}{\load_{\type}}\right\} \right)ds=\numjoblim_{\type}(0)+\int_{0}^{t}F_{\type}\left(\numjoblim_{\type}(s)\right)ds.
\]
Now we compare $\numjobaupscaled_{\type}(t)$ with $\numjoblim_{\type}(t)$:
\begin{align}
&  \left|\numjobaupscaled_{\type}(t)-\numjoblim_{\type}(t)\right| \nonumber\\
\leq & \left|\numjobaupscaled_{\type}(0)-\numjoblim_{\type}(0)\right|+\int_{0}^{t}\left|F_{\type}\left(\numjobaupscaled_{\type}(s)\right)-F_{\type}\left(\numjoblimup_{\type}(t)\right)\right|ds \nonumber\\
 & +\frac{1}{\lambdasup_{\type}}\left|\overline{N}_{\type,1}\left(\int_{0}^{t}\lambdasup_{\type}ds\right)\right|+\frac{1}{\lambdasup_{\type}}\left|\overline{N}_{\type,2}\left(\int_{0}^{t}\lambdasup_{\type}\mu_{\type}\min\left\{ \numjobaupscaled_{\type}(s),\frac{\numserver}{\lambdasup_{\type}\resourceup_{\type}}\right\} ds\right)\right| \nonumber \\
\leq & \left|\numjobaupscaled_{\type}(0)-\numjoblim_{\type}(0)\right|+\int_{0}^{t}\mu_{\type}\left|\numjobaupscaled_{\type}(s)-\numjoblimup_{\type}(t)\right|ds\nonumber \\
 & +\frac{1}{\lambdasup_{\type}}\left|\overline{N}_{\type,1}\left(\int_{0}^{t}\lambdasup_{\type}ds\right)\right|+\frac{1}{\lambdasup_{\type}}\left|\overline{N}_{\type,2}\left(\int_{0}^{t}\lambdasup_{\type}\mu_{\type}\min\left\{ \numjobaupscaled_{\type}(s),\frac{\numserver}{\lambdasup_{\type}\resourceup_{\type}}\right\} ds\right)\right|,\label{eq:centered}
\end{align}
where the second inequality follows from the Lipschitz property of $F_{\type}(\cdot)$.

Next we bound the terms in (\ref{eq:centered}). By applying a classical
inequality on unit-rate Poisson process (see Proposition
5.2 in~\cite{DraMas_09}), we have 
\begin{align*}
\Pr\left(\sup_{0\le t\le T}\frac{1}{\lambdasup_{\type}}\left|\overline{N}_{\type,1}\left(\int_{0}^{t}\lambdasup_{\type}ds\right)\right|>\epsilon'\right) & =\Pr\left(\sup_{0\le t\le\lambdasup_{\type}T}\left|\overline{N}_{\type,1}(t)\right|>\epsilon'\lambdasup_{\type}\right) \le 2e^{-\lambdasup_{\type}T\cdot h\left(\frac{\epsilon'}{T}\right)},
\end{align*}
where $h$ is a function defined by $h(u)=(1+u)\log(1+u)-u$.
Similarly, 
\begin{align*}
 & \Pr\left(\sup_{0\le t\le T}\frac{1}{\lambdasup_{\type}}\left|\overline{N}_{\type,2}\left(\int_{0}^{t}\lambdasup_{\type}\mu_{\type}\min\left\{ \numjobaupscaled_{\type}(s),\frac{\numserver}{\lambdasup_{\type}\resourceup_{\type}}\right\} ds\right)\right|>\epsilon'\right) \\
\leq & \Pr\bigg(\sup_{0\le t\le {\numserver\mu_{\type}T/\resourceup_{\type}}}\left|\overline{N}_{\type,2}\left(t\right)\right|>\epsilon'\lambdasup_{\type}\bigg) 
\leq  2e^{-\frac{\numserver\mu_{\type}T}{\resourceup_{\type}}h\left(\frac{\load_{\type}\epsilon'}{T}\right)}.
\end{align*}

Combining these bounds, we have 
\begin{align*}
 & \Pr\left(\sup_{0\le t\le T}\left(\left|\numjobaupscaled_{\type}(t)-\numjoblim_{\type}(t)\right|-\int_{0}^{t}\mu_{\type}\left|\numjobaupscaled_{\type}(s)-\numjoblim_{\type}(t)\right|ds\right)>3\epsilon'\right)\\
\leq & \Pr\Big(\big|\numjobaupscaled_{\type}(0)-\numjoblim_{\type}(0)\big|>\epsilon' \Big) + 2e^{-\lambdasup_{\type}T\cdot h\left(\frac{\epsilon'}{T}\right)}+2e^{-\frac{\numserver\mu_{\type}T}{\resourceup_{\type}}h\left(\frac{\load_{\type}\epsilon'}{T}\right)}.
\end{align*}
Note that the function $\big|\numjobaupscaled_{\type}(t)-\numjoblim_{\type}(t)\big|$
is finite with probability $1$ on the interval $[0,T].$ Applying
Gronwall's lemma yields
\begin{align}
 & \Pr\left(\sup_{0\le t\le T}\left|\numjobaupscaled_{\type}(t)-\numjoblim_{\type}(t)\right|>\frac{\epsilon}{\numtype}\right)\nonumber \\
\leq & \Pr\left(\sup_{0\le t\le T}\left(\left|\numjobaupscaled_{\type}(t)-\numjoblim_{\type}(t)\right|-\int_{0}^{t}\mu_{\type}\left|\numjobaupscaled_{\type}(s)-\numjoblim_{\type}(t)\right|ds\right)>\frac{\epsilon e^{-\mu_{\type}T}}{\numtype}\right)\nonumber \\
\leq & \Pr\bigg(\big|\numjobaupscaled_{\type}(0)-\numjoblim_{\type}(0)\big|>\frac{\epsilon e^{-\mu_{\type}T}}{3\numtype} \bigg)+ 2e^{-\lambdasup_{\type}T\cdot h\left(\frac{\epsilon e^{-\mu_{\type}T}}{3\numtype T}\right)}+2e^{-\frac{\numserver\mu_{\type}T}{\resourceup_{\type}}h\left(\frac{\load_{\type}\epsilon e^{-\mu_{\type}T}}{3\numtype T}\right)},\label{eq:bound-aux-y-n}
\end{align}
where the last step is obtained by setting $\epsilon'=\frac{\epsilon e^{-\mu_{\type}T}}{3\numtype}$. This completes Step~1.

Next we prove \textbf{Step 2}. As we noted, $\numjobup_{\type}(t)=\numjobaup_{\type}(t)$
for all class $\type$ and all $t\in[0,T]$ if 
$\sup_{0\le t\le T}\sum_{\type=1}^{\numtype}\resourceup_{\type}\numjobaup_{\type}(t)\le\numserver.$
We make the following claim (the proof is provided at the end).
\begin{claim}
The inequality $\sup_{0\le t\le T}\sum_{\type=1}^{\numtype}\resourceup_{\type}\numjobaup_{\type}(t)\le\numserver$ holds true if, for all $\type$,
\begin{equation}
\sup_{0\le t\le T}\left|\numjobaupscaled_{\type}(t)-\numjoblim_{\type}(t)\right|\le \delta, \label{eq:fluid_appro}
\end{equation}
where $\delta=\frac{1}{\mu_{\max}\load}\big(1-\sum_{\type: \numjoblim_\type(0)\geq 1/\mu_\type } \load_{\type}\mu_{\type}\numjoblim_\type(0) - \sum_{\type:  \numjoblim_\type(0) < 1/\mu_\type} \load_\type\big)$ is a positive constant. 
\label{claim:processproof}
\end{claim} 

By Claim~\ref{claim:processproof}, we have 
\begin{align}
 & \Pr\left(\exists \text{ \ensuremath{\type} and \ensuremath{t} $\in [0,T]$ s.t. } \numjobup_{\type}(t)\neq\numjobaup_{\type}(t)\right) \nonumber\\
\leq & \sum_{\type=1}^{\numtype}\Pr\left(\sup_{0\le t\le T}\left|\numjobaupscaled_{\type}(t)-\numjoblim_{\type}(t)\right|>\delta \right)\nonumber\\
\leq & \sum_{\type=1}^{\numtype}\left( \Pr\bigg(\big|\numjobaupscaled_{\type}(0)-\numjoblim_{\type}(0)\big|> \frac{\delta}{3}\bigg)+ 2e^{-\lambdasup_{\type}T\cdot h\left(\frac{\delta e^{-\mu_{\type}T}}{3 T}\right)}+2e^{-\frac{\numserver\mu_{\type}T}{\resourceup_{\type}}h\left(\frac{\delta\load_{\type} e^{-\mu_{\type}T}}{3 T}\right)}\right). \label{eq:bound_deviate}
\end{align}

As the last step, we plug the upper bounds \eqref{eq:bound-aux-y-n} and \eqref{eq:bound_deviate} into \eqref{eq:deviate} and \eqref{eq:term-aux} and obtain:
\begin{align}
 & \Pr\bigg(\sup_{0\le t\le T}\sum_{\type=1}^{\numtype}\left|\scalenumjobup_\type-\numjoblim_{\type}(t)\right|>\epsilon\bigg)\nonumber \\
\le & 2\sum_{\type=1}^{\numtype}\bigg(e^{-\lambdasup_{\type}T\cdot h\left(\frac{\epsilon e^{-\mu_{\type}T}}{3\numtype T}\right)}+e^{-\frac{\numserver\mu_{\type}T}{\resourceup_{\type}}h\left(\frac{\load_{\type}\epsilon e^{-\mu_{\type}T}}{3\numtype T}\right)}+e^{-\lambdasup_{\type}T\cdot h\left(\frac{\delta e^{-\mu_{\type}T}}{3T}\right)}+e^{-\frac{\numserver\mu_{\type}T}{\resourceup_{\type}}h\left(\frac{\load_{\type}\delta e^{-\mu_{\type}T}}{3T}\right)}\bigg)\label{eq:bound_exp}\\
& +\sum_{\type=1}^{\numtype} \left( \Pr\bigg(\big|\numjobaupscaled_{\type}(0)-\numjoblim_{\type}(0)\big|>\frac{\epsilon e^{-\mu_{\type}T}}{3\numtype} \bigg)+ \Pr\bigg(\big|\numjobaupscaled_{\type}(0)-\numjoblim_{\type}(0)\big|> \frac{\delta}{3}\bigg)  \right)
\end{align}

By the assumption that $\numjobaupscaled_\type(0)=\scalenumjobup_\type(0)$ and $\lim_{\numserver \rightarrow\infty} \scalenumjobup_\type (0) =\numjoblim_\type(0)$ in probability, we have 
\begin{align*}
    \lim_{\numserver \rightarrow \infty} \Pr\bigg(\big|\numjobaupscaled_{\type}(0)-\numjoblim_{\type}(0)\big|>\frac{\epsilon e^{-\mu_{\type}T}}{3\numtype} \bigg) = 0,\qquad \lim_{\numserver \rightarrow \infty} \Pr\bigg(\big|\numjobaupscaled_{\type}(0)-\numjoblim_{\type}(0)\big|> \frac{\delta}{3}\bigg) = 0.
\end{align*}
With $\resourcemax=o(n)$, it follows that $\lambdasup_\type =\frac{n\load_\type \mu_\type}{\resourceup_\type} \rightarrow \infty$ as $\numserver \rightarrow \infty,$ thus the terms in \eqref{eq:bound_exp} converge to $0$ as $\numserver\rightarrow \infty.$   Therefore, 
\begin{align*}
   \lim_{\numserver \rightarrow \infty} \Pr\bigg(\sup_{0\le t\le T}\sum_{\type=1}^{\numtype}\left|\scalenumjobup_\type-\numjoblim_{\type}(t)\right|>\epsilon\bigg) =0,\qquad \forall \epsilon>0. 
\end{align*}
This completes the proof of Theorem~\ref{THM:FINITE_TIME}. \qed

\begin{proof}[Proof of Claim~\ref{claim:processproof}]
 By the condition in (\ref{eq:fluid_appro}), we have 
\begin{align*}
\sum_{\type=1}^{\numtype}\resourceup_{\type}\numjobaup_{\type}(t) & =\sum_{\type=1}^{\numtype}\numserver\load_{\type}\mu_{\type}\numjobaupscaled_{\type}(t)
 \leq n\delta \mu_{\max}\load+n\sum_{\type=1}^{\numtype}\load_{\type}\mu_{\type}\numjoblim_{\type}(t).
\end{align*}
Recall that the initial condition $\numjoblimvec(0)$ satisfies 
$\numjoblim_{\type}(0)<\frac{1}{\load\mu_{\type}}<\frac{1}{\load_{\type}\mu_{\type}}$ for all $\type$. We then have
$\numjoblim_{\type}(t)=\left(y_{i}(0)-\frac{1}{\mu_{i}}\right)e^{-\mu_{i}t}+\frac{1}{\mu_{i}},$ for each $\type\in\{1,\ldots,\numtype\}.$ Therefore, for any $t\in [0,T],$ 
\begin{align*}
\sum_{\type=1}^{\numtype}\load_{\type}\mu_{\type}\numjoblim_{\type}(t) & = \sum_{\type=1}^{\numtype}\load_{\type}\mu_{\type}\left[\left(\numjoblim_{\type}(0)-\frac{1}{\mu_{\type}}\right)e^{-\mu_{\type}t}+\frac{1}{\mu_{\type}}\right]  \leq \sum_{\type: \numjoblim_\type(0)\geq 1/\mu_\type } \load_{\type}\mu_{\type}\numjoblim_\type(0) + \sum_{\type:  \numjoblim_\type(0) < 1/\mu_\type} \load_\type \triangleq \alpha.
\end{align*}
Note that $\alpha < \sum_{\type: \numjoblim_\type(0)\geq 1/\mu_\type } \frac{\load_\type}{\load} + \sum_{\type:  \numjoblim_\type(0) < 1/\mu_\type} \frac{\load_\type}{\load}  =1.$ With $\delta=\frac{1-\alpha}{\mu_{\max}\load},$ it follows that
\begin{align*}
\sup_{0\le t\le T}\sum_{\type=1}^{\numtype}\resourceup_{\type}\numjobaup_{\type}(t) & \leq n\delta \mu_{\max}\load+ n\alpha=n.
\end{align*}
This completes the proof of Claim~\ref{claim:processproof}.
\end{proof}

\end{document}